\documentclass[10pt,journal,compsoc]{IEEEtran}
\ifCLASSOPTIONcompsoc
  \usepackage[nocompress]{cite}
\else
  \usepackage{cite}
\fi
\ifCLASSINFOpdf
\else
\fi
\usepackage{amsmath}
\usepackage{color}
\usepackage{graphicx}
\usepackage{threeparttable}
\usepackage[linesnumbered, ruled, vlined]{algorithm2e}
\usepackage{amsmath}
\usepackage{amssymb}
\usepackage{amsthm}
\usepackage{mathtools}
\usepackage{soul}
\usepackage[normalem]{ulem}
\usepackage{times}
\usepackage{graphicx, caption, subcaption}
\usepackage{diagbox}

\usepackage{thmtools}
\usepackage{thm-restate}
\usepackage{cleveref}

\newtheorem{definition}{\textbf{Definition}}[section]
\newtheorem{proposition}{\textbf{Proposition}}[section]
\newtheorem{assumption}{Assumption}
\newtheorem{principle}{\textbf{Principle}}[section]

\usepackage[noend]{algpseudocode}
\usepackage{indentfirst}

\DeclareMathOperator*{\argmin}{argmin}

\makeatletter
\newcommand{\rmnum}[1]{\romannumeral #1}
\newcommand{\Rmnum}[1]{\expandafter\@slowromancap\romannumeral #1@}
\makeatother

\usepackage{array}
\newcolumntype{C}[1]{>{\centering\arraybackslash}m{#1}}

\definecolor{azure(colorwheel)}{rgb}{0.0, 0.5, 1.0}

\definecolor{frenchblue}{rgb}{0.0, 0.45, 0.73}

\definecolor{bittersweet}{rgb}{1.0, 0.44, 0.37}

\definecolor{green(pigment)}{rgb}{0.0, 0.65, 0.31}

\definecolor{navyblue}{rgb}{0.0, 0.0, 0.5}

\definecolor{darkcerulean}{rgb}{0.03, 0.27, 0.49}

\definecolor{darkpowderblue}{rgb}{0.0, 0.2, 0.6}

\definecolor{egyptianblue}{rgb}{0.06, 0.2, 0.65}

\begin{document}
%
\title{Towards Cost-Optimal Policies for DAGs to Utilize IaaS Clouds with Online Learning}
%
%
%
%

\author{Xiaohu~Wu,
        Han Yu,
        Giuliano Casale, and Guanyu~Gao
\IEEEcompsocitemizethanks{\IEEEcompsocthanksitem Xiaohu Wu and Han Yu are with the School of Computer Science and Engineering, Nanyang Technological University, Singapore. 
E-mail: \{xiaohu.wu, han.yu\}@ntu.edu.sg
\IEEEcompsocthanksitem Giuliano Casale is with the Department of Computing, Imperial College London, United Kingdom.
E-mail: g.casale@imperial.ac.uk

\IEEEcompsocthanksitem Guanyu Gao is with the School of Computer Science and Engineering, Nanjing University of Science and Technology, China. 
E-mail: gygao@njust.edu.cn


}
\thanks{Manuscript received April 19, 2005; revised August 26, 2015.}}

%
%

\markboth{Journal of \LaTeX\ Class Files,~Vol.~14, No.~8, August~2015}%
{Shell \MakeLowercase{\textit{{\em et al.}}}: Bare Demo of IEEEtran.cls for Computer Society Journals}
%



\IEEEtitleabstractindextext{%
\begin{abstract}
Premier cloud service providers (CSPs) offer two types of purchase options, namely on-demand and spot instances, with time-varying features in availability and price. Users like startups have to operate on a limited budget and similarly others hope to reduce their costs. While interacting with a CSP, central to their concerns is the process of cost-effectively utilizing different purchase options possibly in addition to self-owned instances. A job in data intensive applications is typically represented by a directed acyclic graph which can further be transformed into a chain of tasks. The key to achieving cost efficiency is determining the allocation of a specific deadline to each task, as well as the allocation of different types of instances to the task. In this paper, we propose a framework that determines the optimal allocation of deadlines to tasks. The framework also features an optimal policy to determine the allocation of spot and on-demand instances in a predefined time window, and a near-optimal policy for allocating self-owned instances. The policies are designed to be parametric to support the usage of online learning to infer the optimal values against the dynamics of cloud markets.
Finally, several intuitive heuristics are used as baselines to validate the cost improvement brought by the proposed solutions. We show that the cost improvement over the state-of-the-art is up to 24.87\% when spot and on-demand instances are considered and up to 59.05\% when self-owned instances are considered.
\end{abstract}

}

\maketitle

\IEEEdisplaynontitleabstractindextext

%
\IEEEpeerreviewmaketitle

\section{Introduction}


The worldwide Infrastructure as a Service (IaaS) cloud market is attracting various users and grew 37.3\% in 2019 to total \$44.5 billion \cite{Gartner}. IaaS enables users to escape purchase and maintenance of servers whose capacity has to satisfy their peak demand to avoid unacceptable latency. Users can scale up or down their computing capacity by renting servers from IaaS providers to match the variation in demand over time. The dominant IaaS providers include Amazon Elastic Cloud Compute (EC2), Microsoft Azure, and Google Cloud, accounting for 45.0\%, 17.9\% and 5.3\% of the global market share respectively. The ongoing COVID-19 pandemic also provides a further push for the adoption of IaaS as more enterprises move their applications to public clouds. To bridge the gap between IaaS providers and users, the key is to determine the process for users to cost-effectively use IaaS services, which enhances user engagement and satisfaction and long-term sustainability of cloud ecosystems \cite{Giuliano19a}.



On-demand and spot instances are two typical purchase options \cite{Wu19a}. On-demand instances are always available at a fixed price once requested by users. Users pay only when instances are actually consumed. Differently from Amazon EC2, spot instances are called spot virtual machines (VMs) in Microsoft Azure \cite{Microsoft-Azure} and preemptive VM instances in Google Cloud \cite{Google-Cloud}. Spot instances have uncertain availability. Generally, CSPs may reclaim the resources of spot instances at any time point for other purposes. In Google Cloud, spot prices are fixed and the instance availability only depends on the dynamics of system resources. In Amazon EC2 and Microsoft Azure, spot prices vary over time and a user needs to bid a price for spot instances; the instance availability also depends on the relation of the spot and bid prices. Spot instances can reduce costs by up to 50-90\% compared to on-demand instances \cite{aws}. On the other hand, a user can have its own instances, called self-owned instances, which, although insufficient at times, can be extended with additional IaaS instances purchased on-demand from the cloud. Also, some users may have no self-owned instances ({\em e.g.}, in the case of startups) and need to buy all necessary computing resources. 


Previous works \cite{Jain14,Wu20a,Jain14P,Wu17} have enabled cost-effectively processing a special type of workloads, namely map-only tasks \cite{Fox11a,Jain12,Wu15a}; each task is partitioned into a large number of independent sub-tasks that can be executed on multiple instances simultaneously; there is also a parallelism bound specifying the maximum number of instances that the task can utilize simultaneously. However, such tasks are independent and can only cover a limited number of important applications. The workload is more generally described by a directed acyclic graph (DAG) whose nodes are tasks and whose edges represent precedence constraints among tasks \cite{Nagarajan13a,Fox11a}; each DAG is referred to as a job. Examples of such jobs include the workloads of MapReduce and Spark's RDDs \cite{Dean08a,Zaharia12a,Verma16a}, which are fundamental programming paradigms for big-data processing. For a user, its jobs arrive over time, each with a specific timing requirement, {\em i.e.}, a deadline by which to complete all its tasks. Each job will be allocated instances of different types (self-owned, on-demand and spot). Our problem is to find an allocation that minimizes cost while meeting the deadline requirement of the job and the precedence constraints among its tasks.


\vspace{0.08em}\noindent\textbf{Challenges.} The costs of self-owned, spot and on-demand instances are increasing. To be cost-optimal, the objective of an allocation policy should be maximizing the utilization of self-owned and then spot instances and minimizing the utilization of costly on-demand instances. One component of our framework is the policy for allocating different types of instances to a single task to be executed in a predefined time window and it involves determining the proportions of different instances. Previous works \cite{Jain14,Wu20a} consider a discrete allocation case where the allocation of spot and on-demand instances is updated on an hourly basis, which arises in a class of instances in Amazon EC2 where the billing of on-demand instances is done on an hourly basis. In this paper, we consider the continuous allocation case with a reformulated analysis; here, users pay for the exact period in which on-demand instances are consumed. The resulting framework applies to the other class of instances in Amazon EC2 and the instances of Microsoft Azure and Google Cloud.

The other new aspect is addressing the precedence constraints among the tasks of a job. A task can be executed only when all its preceding tasks have been finished. For analytical tractability, a DAG job is normally transformed into a job with a chain precedence constraint ({\em i.e.,} a sequence of tasks) where one task can be executed only if its preceding task is completed \cite{Nagarajan13a}. Spot instances are available at irregular intervals. The minimum execution time needed to finish a single task is its workload divided by its parallelism bound. Suppose a user has no self-owned instances. {An intuitive greedy strategy does not work well}: it first requests to {fully} utilize spot instances to finish tasks one by one until some time point after which all remaining tasks have to {fully} utilize costly on-demand instances to meet the job deadline. In contrast, difference exists among tasks and the capacity of a task utilizing spot instances depends on its characteristics and the length of an associated time window in which it is executed. Given a job, its tasks can be executed from its arrival time until its deadline, and a proper allocation of time window sizes to its tasks is needed to maximize the total utilization of spot instances. 



\vspace{0.11em}\noindent\textbf{Our Contributions.} {\em Technically}, the main contribution of this paper is to propose a framework that enables utilizing a class of IaaS services to process jobs with chain precedence constraints cost-effectively, where on-demand instances are charged for the period in which they are exactly consumed:
\begin{itemize}
\item In the case that a single task is to be executed in a time window, we derive policies that allocate spot and on-demand instances cost-optimally and self-owned instances cost-effectively. This is the basis to derive the capacity that a task can achieve to utilize spot instances, given the time window length.

\item A job has multiple tasks. We derive an optimal yet efficient allocation of time window sizes to the tasks, based on a formulation of the problem as an integer linear program to maximize the utilization of spot instances. The allocation algorithm can be used both when the tenant has self-owned resources and when it does not.


\end{itemize}
Leveraging existing techniques in combinatorial optimization, a DAG job can be transformed into a job with a chain precedence constraint \cite{Nagarajan13a}. {\em Consequently}, our technical framework can be used to cost-effectively utilize cloud services for the general DAG jobs. It applies to a significant class of instances in Amazon EC2 and the instances of Microsoft Azure and Google Cloud. Experimentally, several intuitive heuristics are used as baselines to validate the cost improvement brought by the proposed solutions. The cost saving is up to 24.87\% when spot and on-demand instances are considered and up to 59.05\% when self-owned instances are considered. In our framework, the policies and algorithm are parametric, in terms of the availability of spot instances and the sufficiency of self-owned instances. Like \cite{Wu20a}, we leverage the online learning technique of \cite{Jain14,Jain14P} to infer these parameters. {We note that the sufficiency is indicated by a parameter $\beta_{0}$ that controls the allocation of self-owned instances (see Section~\ref{sec.allocation-self-owned}); the more self-owned instances a user has, the smaller the value of $\beta_{0}$ and the more self-owned instances each task gets allocated. While a user requests spot instances, their availability is quantified as the averaged proportion of the period in which spot instances are available.}

The rest of this paper is organized as follows. The related work is introduced in Section~\ref{sec.related-work}. We formally describe the problem in Section~\ref{sec.model}. In Section~\ref{sec.scheduling}, we propose a technical framework for allocating deadlines and self-owned, on-demand and spot instances. In Section~\ref{sec.generalized-online-learning}, we introduce the existing techniques for job transformation and online learning, which will be integrated into our framework. Experimental results are given in Section~\ref{sec.evaluation} to validate the effectiveness of the solutions of this paper. Finally, we conclude this paper in Section~\ref{sec.conclusion}. 

\section{Related Work}
\label{sec.related-work}

To date, multiple service and pricing models have been proposed \cite{Wu20b,Zhang14a} and the spot and on-demand service model is a major service offering \cite{Ludwig19a,Wu19a,Guerin20a}. Jain {\em et al.} are the first to enable the application of an online learning approach to infer the cost-effective parametric policy for utilizing spot and on-demand instances \cite{Jain14,Jain14P}. The key to achieving cost efficiency is the design of a parametric policy and another limitation of \cite{Jain14,Jain14P} is that self-owned instances are not considered. Next, with this approach, Wu {\em et al.} formalize the instance allocation process and derive the expected optimal parametric policy for spot and on-demand instances and the near-optimal parametric policy for self-owned instances \cite{Wu17,Wu20a}. The works \cite{Jain14,Jain14P,Wu17,Wu20a} simply consider the allocation to independent map-only tasks. Also, in their framework, on-demand instances are charged on an hourly basis, and users have to consider maximizing the usage of instances to integer hours to avoid extra charge. In our framework, users pay by the second and thus for what they exactly consume. {Thus, our} policies for a single task have different forms than the policies of \cite{Wu17,Wu20a}. We will also propose an approach to deal with the precedence constraints among tasks. The online learning approach is interesting in that it does not need prior statistical workload characterization, compared to other techniques such as stochastic programming.

Also, there are other works that simply consider independent tasks and associate a specific deadline with each task to make the instance allocation process manageable. Specifically, Zafer {\em et al.} use a Markov model to characterize spot prices and derive an optimal bidding strategy to utilize spot instances \cite{Song12b}. Yao {\em et al.} formulate the problem of utilizing reserved and on-demand instances as an integer program, and propose heuristic algorithms that give approximate solutions \cite{Yao14a}.

Now, we briefly review other approaches to cost-effective use of cloud services. There is one class of works based on priori statistical knowledge of the workload or spot prices. For instance, Hong {\em et al.} and Chaisiri {\em et al.} apply stochastic programming for reserved and on-demand instances \cite{Hong,Chaisiri}; Zheng {\em et al.} derive the optimal bidding strategy for spot instances \cite{Zheng15}. However, the computational cost of deriving the related statistical knowledge is high \cite{Shib}.
Wang {\em et al.} apply the Bahncard problem for reserved and on-demand instances and the resulting algorithm is analyzed by competitive analysis \cite{Wang}. Vintila {\em et al.} propose a genetic algorithm for spot and on-demand instances \cite{Vintila13a}. Shi {\em et al.} apply Lyapunov optimization and are among the first to jointly utilize the three common types of cloud instances \cite{Shib}; yet, a large job delay is incurred \cite{Wu20a}. Gao {\em et al.} consider the joint resource provisioning and task scheduling and propose a two-timescale markov decision process approach to maximize the profit of a multimedia service provider \cite{Gao16a}. Dubois {\em et al.} propose a heuristic to help cloud users decide the right type of spot instances and the bid price, aiming to minimize the cost while maintaining an acceptable level of performance \cite{Dubois15a,Dubois16a}.

\section{Problem Description and Model}
\label{sec.model}

In this section, we introduce the cloud pricing models, define the operational space of a user to utilize various instances, and characterize the objective of this paper.


\subsection{Resource Availability and Pricing Structure}
\label{sec.pricing-model}



On-demand and spot instance services available at popular CSPs may be modelled as follows. {\em First}, the price $p$ of an {\em on-demand instance} is fixed per unit of time and such instances are always available once requested by a user. For example, the price of utilizing an instance for one hour is posted to users that pay for computing capacity by the second. When a user utilizes an instance for $x$ hours, it is charged $p\cdot x$ where $x$ can be fractional. This is more convenient to users compared with other pricing where billing is done on an hourly basis; in the latter, users have to consider maximizing the usage of instances to integer hours to avoid extra charge.

{\em Second}, a user can also request {\em spot instances} at a lower price than on-demand instances. Their availability varies over time and users can only utilize spot instances occasionally. Factors affecting availability include the idleness of cloud systems generally and the bid price in some scenarios. The cloud can reclaim spot instances allocated to a user at any time whenever it needs to access to those resources for other high-priority jobs; the bid price is the maximum price that a user is willing to pay for spot instances. In Google cloud, spot instances are offered at a fixed price; they are delivered to a user when there are idle instances. In Amazon EC2 and Microsoft Azure, the price of spot instances varies over time; a user successfully gets the spot instances only if its bid price exceeds the spot price; spot instances are reclaimed by the cloud when either there are inadequate resources or the bid price is below the spot price. From a user perspective, spot service is a type of stochastic service. When a user persistently requests spot instances, the spot service commences at random time points and lasts for random durations. To facilitate analysis, we let $\beta$ denote the average portion for which a user can spot instances per unit of time where $\beta\in (0, 1)$.

{\em Third}, a user might have its own instances, {\em i.e.}, {\em self-owned instances}, whose amount is limited or zero and denoted by $r$. If any, the (averaged) cost of utilizing self-owned instances is assumed to be the cheapest compared with cloud instances, which implies that a user always prefers to first utilize its own instances before purchasing instances from the cloud. Thus, without loss of generality, this cost is assumed to be zero.


To sum up, availability and price are two key features in cost management. From a user perspective, on-demand instances are always available as if the CSP has infinite on-demand instances to deliver. A user can also request multiple spot instances, which however are available occasionally. The availability of on-demand and spot instances is illustrated in Figure~\ref{Fig-Instance-Availability}. If any, self-owned instances are finite and are always available; however, they may be insufficient at times to satisfy the user computing need. The costs of utilizing self-owned, spot, and on-demand instances are increasing.

\begin{figure}[t!]

        \centering
        \includegraphics[height=0.73in]{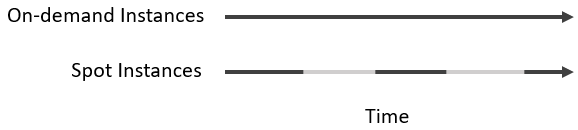}

    \caption{Availability of Cloud Instances: the black ({\em resp.} grey) segments imply the corresponding instances are available ({\em resp.} unavailable) as time goes by.}
    \label{Fig-Instance-Availability}
\end{figure}

%





\subsection{DAG-Structured Jobs}
\label{sec.jobs}

As the time horizon expands, the job arrival of a tenant is monitored at every moment. The tenant plans to rent instances from IaaS clouds to process its jobs and aims to minimize the cost of completing a set of jobs $\mathcal{J}$ (that arrive over a time horizon $T$) by their deadlines. Following \cite{Nagarajan13a,Jain14,Wu20a}, each job $j$ is characterized by a DAG. It has an arrival time $a_{j}$ and a deadline $d_{j}$, that is, job $j$ can be executed and has to be finished in a time window $[a_{j}, d_{j}]$. The main notation of this paper is summarized in Table~\ref{table}. The DAG nodes represent tasks and the directed edges represent precedence relations. {Each DAG job $j$ has $l$ tasks and different jobs may have different values of $l$.} We use $i_{1}\prec i_{2}$ to indicate the execution of task $i_{2}$ can begin only after task $i_{1}$ is completed. Thus, a task $i$ can be executed when all its preceding tasks are completed.

Each task of job $j$ consists of a large number of negligible sub-tasks that are independent and can be executed on multiple instances simultaneously. Completing a task means completing all its sub-tasks. Formally, each task $i$ of job $j$ has a workload $z_{i}$ and an upper bound $\delta_{i}$ of parallelism. While executing task $i$, the number of instances assigned to task $i$ could change over time; the parallelism bound $\delta_{i}$ limits the maximum number of instances that can be used to execute task $i$ simultaneously. $z_{i}$ is the instance time that task $i$ has to consume in order to be finished. For example, suppose $z_{i}=2$; to finish task $i$, it needs to consume one instance for two units of time or two instances for one unit of time. When the task $i$ is always executed on the maximum number $\delta_{i}$ of instances, it has {\em the minimum execution time}, which is denoted by
\begin{align}\label{equa-minimum-e-time}
e_{i}=\frac{z_{i}}{\delta_{i}}.
\end{align}

\begin{table}
	\centering
	\begin{threeparttable}[t]
		\caption{Main Notation}
		\begin{tabular}{| C{1.06cm} | C{6.6cm} |}   
			\hline
			\textbf{Symbol} & \textbf{Explanation} \\ \hline
			
			$\mathcal{J}$ & a set of jobs that arrive over time \\ \hline
			
			$j$ and $a_{j}$ & a job of $\mathcal{J}$ and its arrival time \\ \hline
			
			$d_{j}$ & the deadline: job $j$ must be completed by time $d_{j}$ \\ \hline

			  $l$   & the number of tasks in a job $j$ \\ \hline

			  $i$   & a task $i$ in job $j$ where $i\in \{1, 2, \cdots, l\}$ \\ \hline
			
			$z_{i}$ & the size/workload of task $i$, measured in instance time \\ \hline
			
			$\delta_{i}$ & the parallelism bound, {\em i.e.}, the maximum number of instances that can be simultaneously used by task $i$ \\ \hline

			$e_{i}$ & the minimum execution time of task $i$, {\em i.e.}, $e_{i}=z_{i}/\delta_{i}$ \\ \hline
			
			$\varsigma_{i}$ & the deadline by which task $i$ has to be finished \\ \hline
			
			$\tilde{\varsigma}_{i}$ & the earliest time at which the execution of task $i$ can begin \\ \hline
			
			$\hat{\varsigma}_{i}$ & the time window size of task $i$ where $\hat{\varsigma}_{i}=\varsigma_{i}-\tilde{\varsigma}_{i}$ \\ \hline
			
			$r$ & the number of self-owned instances\\ \hline

			$\tilde{z}_{i}$ & at time $\tilde{\varsigma}_{i}$, the workload of task $i$ to be processed by spot and on-demand instances \\ \hline

			$\tilde{z}_{i}(t)$ & at time $t\in [\tilde{\varsigma}_{i}, \varsigma_{i}]$, the workload of task $i$ to be processed by spot and on-demand instances \\ \hline

			$z_{i}^{o}$ & the workload of task $i$ processed by spot instances \\ \hline

			$\beta$ & the availability of spot instances specifying the average duration of utilizing spot instances in each unit of time   \\ \hline

            $\beta_{0}$ & the sufficiency index of self-owned instances, used to control the allocation of self-owned instances via Equation (\ref{policy-self-owned})  \\ \hline
			
	
            $N(t)$ & the number of self-owned instances currently idle at time $t$ \\ \hline

            $N(t_{1},t_{2})$ & the maximum number of self-owned instances that are always available in $[t_{1}, t_{2}]$, {\em i.e.}, $\min\nolimits_{t\in [t_{1}, t_{2}]}N(t)$ \\ \hline

			$s_{i}$ and $o_{i}$ & the numbers of spot and on-demand instances requested for task $i$  \\ \hline
			
			$r_{i}$ & the number of self-owned instances allocated to task $i$\\ \hline

			$\varsigma_{i}$ & the deadline associated with task $i$\\ \hline

		\end{tabular}
		\label{table}
	\end{threeparttable}
\end{table}



\subsection{Problem Description}

Each job $j$ must be finished in a time window $[a_{j}, d_{j}]$. We first need to determine a time window $[\tilde{\varsigma}_{i}, \varsigma_{i}]$ in which each task $i$ of job $j$ is executed, while respecting the precedence constraints among {the} tasks. $\tilde{\varsigma}_{i}$ is the earliest time at which all its preceding tasks $i^{\prime}$ are finished where $i^{\prime}\prec i$ and at which the execution of $i$ can begin. $\varsigma_{i}$ is the deadline by which task $i$ has to be finished.

\subsubsection{Principled Instance Allocation Process}
\label{sec.instance_allocation_process}

While executed in $[\tilde{\varsigma}_{i}, \varsigma_{i}]$, each task $i$ is assigned $r_{i}$ self-owned instances: $r_{i}\geq 0$ if a user possesses self-owned instances ({\em i.e.,} $r>0$), and $r_{i}=0$ otherwise. The amount of workload processed by self-owned instances is $r_{i}\cdot (\varsigma_{i}-\tilde{\varsigma}_{i})$. 
The remaining workload is to be processed by spot and on-demand instances and its amount is $\tilde{z}_{i}=z_{i}-r_{i}\cdot (\varsigma_{i}-\tilde{\varsigma}_{i})$. If $r_{i}=0$, all the workload of task $i$ will be processed by spot and on-demand instances. From time $\tilde{\varsigma}_{i}$ on, task $i$ requests $o_{i}$ on-demand instances and $s_{i}$ spot instances from the cloud to process the remaining workload where
\begin{align}\label{Constraint-spot-on-demand}
s_{i}+o_{i}=\delta_{i}-r_{i}
\end{align}
to satisfy the parallelism constraint.
While task $i$ is being executed at a time $t\in [\tilde{\varsigma}_{i}, \varsigma_{i}]$, the expected workloads that have been processed by on-demand and spot instances are $o_{i}\cdot (t-\tilde{\varsigma}_{i})$ and $\beta\cdot s_{i}\cdot (t-\tilde{\varsigma}_{i})$ respectively. At time $t$, the remaining workload of $i$ to be processed is denoted by $\tilde{z}_{i}(t)$ whose expected value is as follows:
\begin{align*}
\tilde{z}_{i}(t) = \tilde{z}_{i} - o_{i}\cdot (t-\tilde{\varsigma}_{i}) - \beta\cdot s_{i}\cdot (t-\tilde{\varsigma}_{i}).
\end{align*}
With the parallelism constraint, $\frac{\tilde{z}_{i}(t)}{\delta_{i}-r_{i}}$ is the minimum time needed to finish the remaining $\tilde{z}_{i}(t)$ workload.

\begin{definition}\label{def-flexibility}
For a task $i$ with residual instance time $\tilde{z}_{i}(t)>0$, we say that task $i$ has flexibility to utilize unstable spot instances at a moment $t\in [\tilde{\varsigma}_{i}, \varsigma_{i}]$ when the following condition holds:
\begin{align}\label{equa-flexibility}
\frac{\tilde{z}_{i}(t)}{\delta_{i}-r_{i}} < \varsigma_{i}-t.
\end{align}
\end{definition}

Due to inherent uncertainty within spot service, a task $i$ may reach a state where it has to totally utilize $\delta_{i}-r_{i}$ stable on-demand instances in order to finish by its deadline. Formally, as task $i$ is executed, if there exists some time $t$ satisfying
\begin{align*}
\tilde{z}_{i}(t) = (\varsigma_{i}-t)\cdot (\delta_{i}-r_{i}),
\end{align*}
we call such time $t$ as {\em a turning point} and denote it by $\varsigma_{i}^{c}$. At time $\varsigma_{i}^{c}$, we have to give up utilizing cheap spot instances to finish the remaining $\tilde{z}_{i}(\varsigma_{i}^{c})$ workload by the deadline $\varsigma_{i}$. The instance allocation process may have two phases defined below:

\begin{definition}\label{def-allocation-process}
If the turning point exists and $\varsigma_{i}^{c}\neq \tilde{\varsigma}_{i}$, the instance allocation process has two phases:
\begin{itemize}
\item [(\rmnum{1})] $o_{i}$ on-demand and $s_{i}$ spot instances are requested in the period $[\tilde{\varsigma}_{i}, \varsigma_{i}^{c}]$;
\item [(\rmnum{2})] $\delta_{i}-r_{i}$ on-demand instances are utilized in $[\varsigma_{i}^{c}, \varsigma_{i}]$.
\end{itemize}
If the turning point exists and $\varsigma_{i}^{c}=\tilde{\varsigma}_{i}$, $\delta_{i}-r_{i}$ on-demand instances are utilized in the period $[\tilde{\varsigma}_{i}, \varsigma_{i}]$. If the turning point does not exist, $o_{i}$ on-demand instances and $s_{i}$ spot instances are requested until some time $t\in [\tilde{\varsigma}_{i}, \varsigma_{i}]$ such that $\tilde{z}_{i}(t)=0$.
\end{definition}

\vspace{0.1em}\noindent\textbf{Example.} Now, we give a toy example to illustrate the instance allocation process in Definition~\ref{def-allocation-process}. Suppose task $i$ has a parallelism bound $\delta_{i}=3$ and is executed in $[\tilde{\varsigma}_{i}, \varsigma_{i}]=[0, 2]$; the user has $r=1$ self-owned instance. The availability of spot instances is $\beta=0.5$. The scheduler allocates $r_{i}=1$ self-owned instance to task $i$ in $[0,2]$. The remaining workload to be processed is $\tilde{z}_{i}(0)=z_{i}-1\times 2$. The scheduler begins to request one spot and one on-demand instance at time 0 where $o_{i}=s_{i}=1$:
\begin{itemize}
\item If $z_{i}=3.5$, we have $\tilde{z}_{i}(0) = 1.5$. At time 1, task $i$ gets enough execution time from spot and on-demand instances; we thus have $\tilde{z}_{i}(1)=0$ and the turning point does not exist. This is illustrated in Fig.~2(a).
\item If $z_{i}=5.5$, we have $\tilde{z}_{i}(0) = 3.5$. At time 1, the remaining workload is $\tilde{z}_{i}(1)=3.5-1.5=2$; since $\tilde{z}_{i}(1)=(\varsigma_{i}-1)\cdot (\delta_{i}-r_{i})$, the turning point exists and $\varsigma_{i}^{c}=1$; task $i$ need turn to totally utilize $\delta_{i}-r_{i}=2$ on-demand instances to meet the deadline. This is illustrated in Fig.~2(b).
\end{itemize}

\begin{figure}[t!]
    \begin{subfigure}[t]{0.45\columnwidth}
        \centering
        \includegraphics[height=0.80in]{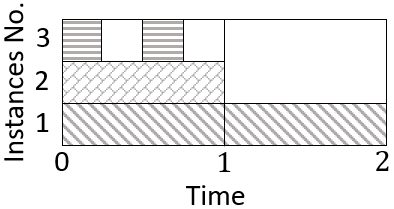}
        \caption{$z_{i}=3.5$}
        \label{Fig-Instance-Allocation-Process-a}
    \end{subfigure}%
    \hspace{0.2em}\begin{subfigure}[t]{0.45\columnwidth}
        \centering
        \includegraphics[height=0.80in]{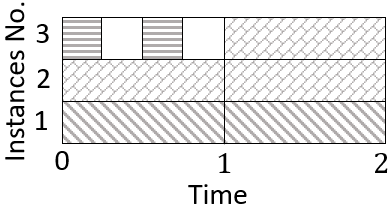}
        \caption{$z_{i}=5.5$}
        \label{Fig-Instance-Allocation-Process-b}
    \end{subfigure}
    \caption{Both plots illustrate the instance allocation of a task: the diagonal stripe, diagonal brick, and horizonal stripes areas denote the workloads processed by self-owned, on-demand and spot instances respectively.}
    \label{Fig-Instance-Allocation-Process}
\end{figure}

\subsubsection{Decision Variables, and Objectives}

Jobs arrive over time. Each job is represented as a DAG and has multiple tasks.

\vspace{0.2em}\noindent\textbf{Decision Variables.} Given a job $j$, we need to determine (\rmnum{1}) the deadline $\varsigma_{i}$ by which each task $i$ is finished and (\rmnum{2}) the numbers of spot and on-demand instances requested when there is flexibility for task $i$ to utilize spot instances ({\em i.e.,} when the turning point has not appeared; see the first and third cases of Definition~\ref{def-allocation-process}). If a user possesses self-owned instances, we also need to determine its amount allocated to each task $i$. Thus, our decision variables include $\varsigma_{i}$, $s_{i}$, $o_{i}$ and $r_{i}$ for each task $i$ of a job.

\vspace{0.2em}\noindent\textbf{Objective of Instance Allocation.} We refer to the ratio of the total cost of utilizing a certain type of instances to the total workload processed by this type of instances as the average unit cost of this type of instances. As described in Section~\ref{sec.pricing-model}, we assume like \cite{Wu17,Wu20a} that

\begin{assumption}\label{assump-order}
The average unit cost of self-owned instances is lower than the average unit cost of spot instances, which is lower than that of on-demand instances. 
\end{assumption}

Due to Assumption~\ref{assump-order}, the overall objective of our instance allocation framework is to maximize the utilization of self-owned and then spot instances and minimize the utilization of costly on-demand instances. Achieving this objective involves properly determining the decision variables $\varsigma_{i}$, $s_{i}$, $o_{i}$ and $r_{i}$ for each task $i$ of a job.
The deadline $\varsigma_{i}$ of a task $i$ affects its instance allocation process by Definition~\ref{def-allocation-process} and thus its completion time; the later affects the time that other tasks can start being executed due to the precedence constraint.

While allocating various instances to a single task $i$ in a specific time window  $[\tilde{\varsigma}_{i}, \varsigma_{i}]$, we should consider allocating various instances to a task in the order of self-owned, spot and on-demand instances; the objectives here are the same as the ones in \cite{Wu17,Wu20a} where only the allocation to a single task is considered. Differently, we consider the case of a DAG job where a user pays exactly for what it consumes. Now, we describe these objectives in Principles~\ref{prin-one} and~\ref{prin-two}.


\begin{principle}\label{prin-one}
If a user possesses self-owned instances, the scheduler should make self-owned instances (\rmnum{1}) fully utilized, and (\rmnum{2}) utilized in a way so as to maximize the opportunity that all tasks have to utilize spot instances.
\end{principle}

\begin{principle}\label{prin-two}
After self-owned instances are used or {if a user has no self-owned instances}, the scheduler should utilize on-demand instances in a way so as to maximize the opportunity that a task has to utilize spot instances.
\end{principle}

Realizing the above principles involves properly determining the decision variables $s_{i}$, $o_{i}$ and $r_{i}$ for each individual task $i$ of a job. Last but not least, a job $j$ has $l$ tasks and has to be finished in a given time window $[a_{j}, d_{j}]$. We also need to maximize the aggregate utilization of self-owned and spot instances by all tasks within the job. Correspondingly, we need to realize the following objective.

\begin{principle}\label{prin-three}
Before allocating instances to the $l$ tasks of a job, the scheduler needs to properly determine the deadlines $\varsigma_{1}, \varsigma_{2}, \cdots, \varsigma_{l}$ to maximize the overall utilization of self-owned instances, if any, and spot instances.
\end{principle}

In the following, we will propose solutions for realizing the three principles above. The final result is an integrated framework for a user to cost-effectively process DAG jobs by renting typical cloud instances from major IaaS providers.

\section{(Near-)Optimal Instance Allocation}
\label{sec.scheduling}

In this section, we consider a special case of jobs, {\em i.e.}, each job is a chain of $l$ tasks, where for all $i\in [2, l]$ the execution of the $i$-th task can begin {\em if and only if} the first $l-1$ tasks have been finished. We propose a framework to design (near-)optimal parametric policies that can effectively realize Principles~\ref{prin-one}-\ref{prin-three}. In the next section, we will use the technique of \cite{Nagarajan13a} to extend the framework to the case where the precedence constraints are present in a general DAG.

\subsection{Spot and On-demand Instances}
\label{sec.case-spot}

In this subsection, we consider the case that a user has no self-owned instances. We will derive a couple of optimal parametric policies in terms of the availability $\beta$ of spot instances to maximize the utilization of spot instances and realize Principle~\ref{prin-two} and~\ref{prin-three} optimally.



\subsubsection{Preliminaries}
\label{sec.example-1}

Consider a job $j$ with a chain of $l$ tasks to be processed in a time window $[a_{j}, d_{j}]$. While processing these tasks, {\em one question} is what deadline $\varsigma_{i}$ is associated to each task $i$ to ensure that the latter tasks have a large enough window $[\varsigma_{i}, d_{j}]$ in which they are finished. {For all $i\in [1, l]$, it is expected that $\varsigma_{i}$ is also the time point at which task $i$ is finished.} To respect the precedence constraints among the tasks, the execution of the $i$-th task can begin when the ($i-1$)-th task is finished {where $\tilde{\varsigma}_{i}=\varsigma_{i-1}$ for all $i\in [2, l]$}. Thus, task $i$ is expected to be executed in $[\varsigma_{i-1}, \varsigma_{i}]$ where $\varsigma_{0}=a_{j}$ trivially and we have
\begin{align}
a_{j}=\varsigma_{0} < \varsigma_{1} < \cdots  < \varsigma_{l} \leq d_{j}. \label{equa-deadline-order}
\end{align}
{\em The other question} is that, given the time window $[\varsigma_{i-1}, \varsigma_{i}]$ of task $i$, what is the optimal composition of instance types ({\em i.e.}, the values of $o_{i}$ and $s_{i}$) to maximize the amount of workload to be processed by spot instances.

For example, let us consider a job $j$ of $l=4$ tasks with $[a_{j}, d_{j}]$ $=$ $[0, 4]$. The task sizes are $z_{1}=1.5$, $z_{2}=0.5$, $z_{3}=2.5$, and $z_{4}$ $=$ $0.5$. The parallelism bounds are $\delta_{1}=2$, $\delta_{2}=1$, $\delta_{3}$ $=$ $3$, and $\delta_{4}$ $=$ $1$. The availability of spot instances is specified as $\beta$ $=0.5$.
We artificially set the deadline $\varsigma_{i}$ of the $i$-th task to $i$ where $i\in [1, 4]$.
Each task $i$ is finished at time point $i$. In this setting, the amount of workload processed by spot instances is 2, which is illustrated in Fig.~\ref{Fig-Instance-Allocation-Process}. However, as seen later, the optimal amount of workload processed by spot instances is $\frac{22}{6}$ by properly setting the values of $\varsigma_{i}$, $o_{i}$, and $s_{i}$ for $i\in [1, 4]$. In the rest of this subsection, for an arbitrary job $j$, we will derive a computationally efficient yet optimal allocation of deadlines $\varsigma_{1}$, $\varsigma_{2}$, $\cdots$, $\varsigma_{l}$ to its tasks. Additionally, we also derive the expected optimal composition of instance types to finish each task $i$, which is in fact one enabler of the optimal deadline allocation. In this subsection, we have
\begin{align}\label{equa-parallelism-no-self-owned}
o_{i}+s_{i} = \delta_{i}
\end{align}
and the number of self-owned instances assigned to each task $i$ is zero, {\em i.e.,} $r_{i}=0$.

\begin{figure}[t!]
        \centering
        \includegraphics[height=0.78in]{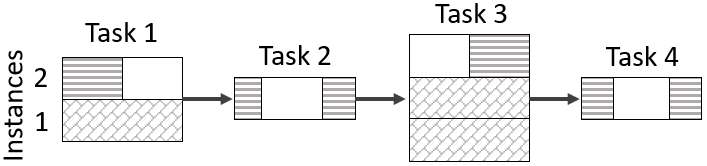}
    \caption{Processing a Chain of 4 Tasks: the diagonal brick, and horizonal stripes areas denote the workloads processed by on-demand and spot instances respectively; in the blank areas, no workload is processed.}
    \label{Fig-Instance-Allocation-Process}
\end{figure}

\subsubsection{Allocation to a Single Task}
\label{sec.spot-policy}

Suppose that the deadlines $\varsigma_{1}$, $\varsigma_{2}$, $\cdots$, $\varsigma_{l}$ are given in advance. In this subsection, we give the expected optimal composition of instance types for a single task $i$ to utilize spot and on-demand instances in the predefined time window $[\varsigma_{i-1}, \varsigma_{i}]$ where $i\in [1, l]$. This will realize Principle~\ref{prin-two} optimally.


The instance allocation process is described in Definition~\ref{def-allocation-process}. Now, we give a condition under which task $i$ can be expected to be finished by utilizing spot instances alone, without utilizing costly on-demand instances. We also derive the expected optimal strategy for task $i$ to utilize different types of instances.

\begin{proposition}\label{proposi-spot}
A task $i$ can be finished by utilizing spot instances alone when the time window size $\hat{\varsigma}_{i}$ satisfies the following condition:
\begin{align}\label{condition-1}
\hat{\varsigma}_{i} = \varsigma_{i}-\varsigma_{i-1} \geq \frac{e_{i}}{\beta}.
\end{align}
The expected optimal strategy of utilizing spot and on-demand instances is as follows:
\begin{itemize}
\item if the condition (\ref{condition-1}) holds, then it is expected that the turning point does not exist and we have $s_{i}=\delta_{i}$ and $o_{i}=0$;
\item if $\hat{\varsigma}_{i} \in \left(e_{i}, \frac{e_{i}}{\beta}\right)$, then the instance allocation process is expected to have two phases and we have in the first phase that $s_{i}=\delta_{i}$ and $o_{i}=0$;
\item if $\hat{\varsigma}_{i} = e_{i}$, then it is expected that the turning point $\varsigma_{i}^{c}$ is $\varsigma_{i-1}$ and we have $o_{i}=\delta_{i}$ and $s_{i}=0$.
\end{itemize}
\end{proposition}
\begin{proof}
See Appendix~\ref{sec.spot}.
\end{proof}

\subsubsection{Optimal Deadline Allocation}

In this subsection, we will realize Principle~\ref{prin-three} optimally. A job $j$ should be executed in the time window $[a_{j}, d_{j}]$.
Our question is finding an optimal allocation of $\varsigma_{1}$, $\varsigma_{2}$, $\cdots$, $\varsigma_{l}$ to maximize the utilization of cheap spot instances and minimize the consumption of costly on-demand instances.

\vspace{0.4em}\noindent\textbf{Formulation as an Integer Linear Program.} We formulate the deadline allocation problem as an integer linear program. To ensure that each task $i$ can be finished in its time window $[\varsigma_{i-1}, \varsigma_{i}]$, we have
\begin{align}
\hat{\varsigma}_{i} = \varsigma_{i} - \varsigma_{i-1} \geq e_{i} \enskip\text{ for all } i\in [1, l] \label{equa-relative-deadline}
\end{align}
where $e_{i}$ is the minimum execution time of task $i$ by (\ref{equa-minimum-e-time}). $\hat{\varsigma}_{i}$ can be written:
\begin{align}\label{equa-window-structure}
\hat{\varsigma}_{i} = e_{i}+x_{i},
\end{align}
where $x_{i}\geq 0$.

With the strategies in Proposition \ref{proposi-spot}, the total amount of workload processed by spot instances has the following relation with the time window size $\hat{\varsigma}_{i}$.

\begin{proposition}\label{proposi-spot-workload}
Given the time window size $\hat{\varsigma}_{i}$, the expected amount of workload processed by spot instances is
\begin{align}\label{equa-spot-workload-alone}
z_{i}^{o} =
\begin{cases}
& \frac{\beta}{1-\beta}\cdot \delta_{i}\cdot x_{i}\enskip\enskip\enskip\, \text{ if } \hat{\varsigma}_{i}\in \left[e_{i}, \frac{e_{i}}{\beta} \right] \\
\vspace{0.24em}& \enskip\enskip\enskip\enskip z_{i} \enskip\enskip\enskip\enskip\enskip\enskip\enskip\enskip \text{ if } \hat{\varsigma}_{i}\in \left[\frac{e_{i}}{\beta}, \infty \right)
\end{cases}
\end{align}
\end{proposition}
\begin{proof}
See Appendix~\ref{sec.spot-workload}.
\end{proof}

Here, $z_{i}^{o}=\frac{\beta}{1-\beta}\cdot \delta_{i}\cdot x_{i}\in [0, z_{i}]$ when $\hat{\varsigma}_{i}\in \left[e_{i}, \frac{e_{i}}{\beta} \right]$. For each task $i$, Proposition~\ref{proposi-spot} and~\ref{proposi-spot-workload} show (\rmnum{1}) the minimum time window size needed to finish the task by only utilizing spot instances, and (\rmnum{2}) how the amount $z_{i}^{o}$ of workload processed by spot instances varies with the time window size and job characteristics.

Our objective is finding an allocation of deadlines $\varsigma_{1}$, $\varsigma_{2}$, $\cdots$, $\varsigma_{h}$ to maximize the utilization of spot instances. This is formulated as an integer linear program below:
\begin{align}\label{equa-linear-program}
\text{maximize} \sum\limits_{k=1}^{h}{z_{i}^{o}},
\end{align}
where $\varsigma_{1}$, $\varsigma_{2}$, $\cdots$, $\varsigma_{h}$ satisfy (\ref{equa-deadline-order}), (\ref{equa-relative-deadline}) and (\ref{equa-window-structure}) and $z_{i}^{o}$ satisfies (\ref{equa-spot-workload-alone}).

\begin{algorithm}[t]
\SetKwInOut{Input}{Input}
\SetKwInOut{Output}{Output}

\Input{the availability $\beta$ of spot services, or the sufficiency index $\beta_{0}$ of self-owned instances ({\em $x=\beta$ or $\beta_{0}$})}
\Output{the time window sizes allocated to the $l$ tasks: $\hat{\varsigma}_{1}$, $\hat{\varsigma}_{2}$, $\cdots$, $\hat{\varsigma}_{l}$}


$\hat{\varsigma}_{i} \leftarrow e_{i}$ for all $i\in [1, l]$\;

$\omega \leftarrow d_{j}-\sum_{i=1}^{l}{e_{i}}$\;

\For{$k\leftarrow 1$ \KwTo $l$}{

    \If{$\omega> \frac{e_{i_{k}}}{x} - \hat{\varsigma}_{i_{k}}$}{

        $\hat{\omega}\leftarrow \frac{e_{i_{k}}}{x} - \hat{\varsigma}_{i_{k}}$,\,
        $\hat{\varsigma}_{i_{k}} \leftarrow \hat{\varsigma}_{i_{k}}+\hat{\omega}$,\,
        $\omega\leftarrow \omega - \hat{\omega}$\;
    }

    \If{$0 < \omega \leq \frac{e_{i_{k}}}{x} - \hat{\varsigma}_{i_{k}}$}{

        $\hat{\varsigma}_{i_{k}} \leftarrow \hat{\varsigma}_{i_{k}}+\omega$,\,
        $\omega\leftarrow 0$\;
    }

}

\caption{Dealloc($x$)}\label{Algo-deadline-allocation-basic}
\end{algorithm}

\vspace{0.4em}\noindent\textbf{Solution.} 
Now, we derive a computationally efficient yet optimal solution to the integer linear program (\ref{equa-linear-program}). By Proposition~\ref{proposi-spot-workload}, we have the following observation. While the time window size $\hat{\varsigma}_{i}$ ranges in $[e_{i}, \frac{e_{i}}{\beta} ]$, the workload $z_{i}^{o}$ of task $i$ processed by spot instances is linearly proportional to $x_{i}$; the larger the parallelism bound $\delta_{i}$, the larger the value of $z_{i}^{o}$. While $\hat{\varsigma}_{i}$ exceeds $\frac{e_{i}}{\beta}$, the workload $z_{i}^{o}$ will not increase any more. We can thus propose a greedy strategy to optimally determine the allocation of deadlines to tasks, {which is presented in Algorithm~\ref{Algo-deadline-allocation-basic} with $\beta$ as an input {\em i.e.,} Dealloc($\beta$)}. Algorithm~\ref{Algo-deadline-allocation-basic} gives the optimal values of $\hat{\varsigma}_{1}, \hat{\varsigma}_{2}, \cdots, \hat{\varsigma}_{l}$, and we can thus derive the optimal values of $\varsigma_{1}$, $\varsigma_{2}$, $\cdots$, $\varsigma_{l}$ by (\ref{equa-deadline-order}) and (\ref{equa-relative-deadline}).

The idea of Dealloc($\beta$) is as follows. Let $\{i_{1}, i_{2}$, $\cdots, i_{l}\}=\{1, 2, \cdots$, $l\}$ be such that $\delta_{i_{1}} \geq \delta_{i_{2}} \geq \cdots\geq \delta_{i_{l}}$. It considers tasks in non-increasing order of their parallelism bounds (line 3) and allocates as much time as possible to the tasks with the largest parallelism bounds. Specifically,
\begin{itemize}
\item Each task $i$ is initially allocated a time window of size $\hat{\varsigma}_{i}^{\ast}=e_{i}$ to guarantee that it can be finished in the allocated window (line 1).
\item The remaining time $\omega = (d_{j}-a_{j}) - \sum_{i=1}^{l}{e_{i}}$ is allocated to the first $l^{\ast}$ tasks with the largest parallelism bounds where $l^{\ast}\in [1, l]$.
    \begin{itemize}
    \item if $l^{\ast}\geq 2$, we have for $k\in [1, l^{\ast}-1]$ that the task $i_{k}$ has a time window size $\hat{\varsigma}_{i_{k}}^{\ast}=\frac{e_{i_{k}}}{\beta}$ (lines 4-5) and the task $i_{l^{\ast}}$ has a time window size $\hat{\varsigma}_{i_{l^{\ast}}}^{\ast}$ $=$ $e_{i_{l^{\ast}}}$ $+$ $(\omega-\sum_{k=1}^{l^{\ast}-1}{(\hat{\varsigma}_{i_{k}}^{\ast}-e_{i_{k}})})$ (lines 6-7);
    \item if $l^{\ast}=1$, the first task has a time window size $\hat{\varsigma}_{i_{1}}^{\ast}$ $=$ $e_{i_{1}}$ $+$ $\omega$ (lines 6-7).
    \end{itemize}
\end{itemize}

\begin{proposition}\label{proposi-optimal-deadline-allocation}
Algorithm~\ref{Algo-deadline-allocation-basic} gives an optimal solution to the integer linear program (\ref{equa-linear-program}) with a time complexity of $\mathcal{O}(n)$.
\end{proposition}
\begin{proof}
See Appendix~\ref{sec.optimal-deadline-allocation}.
\end{proof}

\vspace{0.01em}\noindent\textbf{Example.} Now, we continue the example in Section~\ref{sec.example-1} and show by Algorithm~\ref{Algo-deadline-allocation-basic} and Proposition~\ref{proposi-spot} the expected optimal deadline and instance allocation to the job $j$, which is illustrated in Fig.~\ref{Fig-Optimal-Instance-Allocation-Process}. The optimal deadline allocation is as follows: $\varsigma_{1}=\frac{4}{3}$, $\varsigma_{2}=0.5$, $\varsigma_{3}=\frac{5}{3}$, and $\varsigma_{4}=0.5$. The first task requests two spot instances in $\left[0, \frac{7}{6}\right]$ in the first phase of allocation and two on-demand instances in $\left[\frac{7}{6}, \frac{4}{3}\right]$ in the second phase of allocation; the second task simply requests one on-demand instance in $\left[\frac{4}{3}, \frac{11}{6}\right]$; the third requests three spot instances in $\left[\frac{11}{6}, \frac{21}{6}\right]$; the fourth requests one on-demand instance in $\left[\frac{21}{6}, 4\right]$. Finally, the amount of workload processed by spot instances is $\frac{22}{6}$.

\begin{figure}[t]
        \centering
        \includegraphics[height=0.78in]{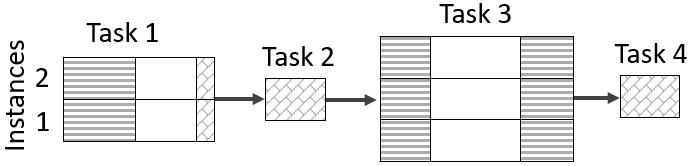}

    \caption{Optimal Processing of a Chain of Tasks: the diagonal brick, and horizonal stripe areas denote the workloads processed by on-demand and spot instances respectively.}
    \label{Fig-Optimal-Instance-Allocation-Process}
\end{figure}

\subsection{Incorporating Self-Owned Instances}

In this section, we extend the framework of Section~\ref{sec.case-spot} to the case with additional self-owned instances.

\subsubsection{Allocation of Self-owned Instances}
\label{sec.allocation-self-owned}

In this subsubsection, we consider the allocation of self-owned instances to a single task $i$ to be finished in $[\varsigma_{i-1},$ $\varsigma_{i}]$. We will give a policy that realizes Principle~\ref{prin-one} effectively. Specifically, the policy needs to guarantee that (\rmnum{1}) self-owned instances are fully utilized by tasks and (\rmnum{2}) in the meantime, the overall opportunity of all tasks utilizing spot instances is maximized. Like \cite{Jain14,Wu20a}, in the subsequent analysis, the issue of rounding the allocations of a job to integers is ignored temporarily for simplicity; in reality, we can round {up} the allocations to integers, {without affecting} the effectiveness of our conclusions much as shown by our experiments.

We will use a common parameter $\beta_{0}$ to determine the amount $r_{i}$ of self-owned instances allocated to each task $i$. $r_{i}$ is defined by a function $f(x)$ when $x=\beta_{0}$. The function $f(x)$ relates to the characteristics of task $i$ and is defined as follows:
\begin{align}\label{equa-f-x}
 f(x) = \max  \left\{  \frac{z_{i}-\delta_{i}\cdot \hat{\varsigma}_{i}\cdot x}{\hat{\varsigma}_{i}\cdot (1-x)},\,\, 0 \right\}.
\end{align}
When $x=0$, $f(x)=\frac{z_{i}}{\hat{\varsigma}_{i}}$; when $x\geq \frac{e_{i}}{\hat{\varsigma}_{i}}$, $f(x)=0$, where $e_{i}$ is given in (\ref{equa-minimum-e-time}). The value of $f(x)$ ranges in $[0, \frac{z_{i}}{\hat{\varsigma}_{i}}]$. We refer to the parameter $\beta_{0}$ as the {\em sufficiency index} of self-owned instances. As we will see, given a set of jobs arriving over time, the value of $\beta_{0}$ is small if self-owned instances are sufficient and large otherwise. 

\begin{proposition}\label{property-self-owned}
The function $f(x)$ has the following properties:
\begin{itemize}
\item $f(\beta)$ is the minimum number such that, after task $i$ is allocated $f(\beta)$ self-owned instances, it is expected that task $i$ can be finished in $[\varsigma_{i-1},$ $\varsigma_{i}]$ by only requesting to utilize $\delta_{i}-f(\beta)$ spot instances without utilizing costly on-demand instances.
\vspace{0.18em}\item $f(x)$ is non-increasing in $x$.
\end{itemize}
\end{proposition}
\begin{proof}
See Appendix~\ref{sec.property-self-owned}.
\end{proof}

Now, we introduce the policy. Let $N(t)$ denote the number of self-owned instances available at time $t$ and $N(t_{1}, t_{2})$ be the maximum number of self-owned instances that are available in the entire time interval $[t_{1}, t_{2}]$, {\em i.e.}, $$N(t_{1}, t_{2})=\min\limits_{t\in [t_{1}, t_{2}]}{N(t)}.$$ The number $r_{i}$ of self-owned instances allocated to task $i$ is defined as follows:
\begin{align}\label{policy-self-owned}
r_{i} = \min \{ f(\beta_{0}),\,  N(\varsigma_{i-1}, \varsigma_{i}), \delta_{i}\}.
\end{align}
Task $i$ can use these instances in the period $[\varsigma_{i-1}, \varsigma_{i}]$.

We show by Proposition~\ref{property-self-owned} that the policy (\ref{policy-self-owned}) can effectively realize Principle~\ref{prin-one}.
$\beta$ represents the availability of spot instances. $f(x)$ is non-increasing in $x$. {\em In the case} that sufficient self-owned instances are available, we can set $\beta_{0}$ to a value smaller than $\beta$ and each task $i$ is assigned more than $f(\beta)$ self-owned instances;
as a result, all tasks can be expected to be finished by utilizing spot instances alone, without consuming costly on-demand instances. In the meantime, by setting $\beta_{0}$ to a properly small value, we can guarantee that self-owned instances are fully utilized by allocating a large number of self-owned instances to each task.

{\em In the case} that self-owned instances are insufficient, we can set $\beta_{0}$ to a value larger than $\beta$, and each task $i$ is assigned less than $f(\beta)$ self-owned instances;
here, all tasks are expected to consume some costly on-demand instances. No tasks are assigned more than $f(\beta)$ self-owned instances. Allocating more than $f(\beta)$ self-owned instances to a task $i$ can lead to a waste of self-owned instances since they can be allocated to other tasks for processing the workload that will otherwise be processed by costly on-demand instances. Finally, as shown in the two cases above, if the policy (\ref{policy-self-owned}) is used, no tasks are overly allocated and a balanced-allocation is achieved to well realize Principle~\ref{prin-one}. This will further be validated in our third experiment of Section~\ref{sec.evaluation}.

\subsubsection{Deadline Allocation}

In the last subsubsection, we have given an explicit form of the policy for self-owned instances. Built on such a policy, we derive in this subsubsection the expected optimal allocation of deadlines under some mild assumptions. 

A task $i$ is allocated to utilize $r_{i}$ self-owned instances in $[\varsigma_{i-1},$ $\varsigma_{i}]$. Task $i$ is divisible and afterwards task $i$ can be viewed as a new task with a parallelism bound $\tilde{\delta}_{i}=\delta_{i}-r_{i}$ and a (remaining) workload/size $\tilde{z}_{i}=z_{i}$ $-$ $r_{i}\cdot \hat{\varsigma}_{i}$, which will be processed by spot and on-demand instances alone. The number $r_{i}$ is defined in (\ref{policy-self-owned}). Within the parallelism bound, it is the minimum of $f(\beta_{0})$ and the maximum number $N(\varsigma_{i-1}, \varsigma_{i})$ of self-owned instances available in $[\varsigma_{i-1}, \varsigma_{i}]$. When a CSP has sufficient self-owned instances, $\beta_{0}$ is set to a value smaller than $\beta$ and a task $i$ is expected to be assigned more than $f(\beta)$ self-owned instances. When a CSP has insufficient self-owned instances, $\beta_{0}$ is set to a value larger than $\beta$ and a task $i$ is expected to be assigned less than $f(\beta)$ self-owned instances.

In any case, by choosing a properly large or small value for $\beta_{0}$, $r_{i}$ can equal or be close to $f(\beta_{0})$. Thus, for analytical tractability, we assume that each task $i$ is assigned $r_{i}=f(\beta_{0})$ self-owned tasks to be utilized in $[\varsigma_{i-1}, \varsigma_{i}]${, although the policy that is actually used in our framework is defined by (\ref{policy-self-owned})}. This helps obtain an informed policy to allocate deadlines to the tasks of a job. The effectiveness of {the resulting policy} will further be validated by our experiments (see Experiments 2 and 3 in Section~\ref{sec.evaluation-results}). Each job is assigned a specific $\beta_{0}$. Depending on the relation between the availability $\beta$ of spot instances and the sufficiency index $\beta_{0}$ of self-owned instances, we have the following conclusion on the amount $z_{i}^{o}$ of workload processed by spot instances after each task $i$ is allocated $r_{i}$ self-owned instances.

\begin{proposition}\label{proposi-workload-self-owned}
Depending on the time window size $\hat{\varsigma}_{i}=e_{i}+x_{i}$, in the case that $\beta_{0}\leq \beta$, we have
\begin{align}
z_{i}^{o} =
\begin{cases}
& \frac{\beta_{0}}{1-\beta_{0}}\cdot \delta_{i}\cdot x_{i}   \enskip\enskip\enskip  \text{ if } \hat{\varsigma}_{i}\in \left[e_{i},\, \frac{e_{i}}{\beta_{0}}\right] \\
& \enskip\enskip\enskip\enskip  z_{i} \enskip\enskip\enskip\enskip\enskip\enskip\enskip\enskip   \text{ if } \hat{\varsigma}_{i} > \frac{e_{i}}{\beta_{0}}.
\end{cases}
\end{align}
In the case that $\beta < \beta_{0}$, we have
\begin{align}
z_{i}^{o} =
\begin{cases}
& \frac{\beta}{1-\beta}\cdot \delta_{i}\cdot x_{i}   \enskip\enskip\enskip  \text{ if } \hat{\varsigma}_{i}\in \left[e_{i},\, \frac{e_{i}}{\beta}\right] \\
& \enskip\enskip\enskip\enskip z_{i} \,\enskip\enskip\enskip\enskip\enskip\enskip\enskip    \text{ if } \hat{\varsigma}_{i} > \frac{e_{i}}{\beta}.
\end{cases}
\end{align}
\end{proposition}
\begin{proof}
See Appendix~\ref{sec.workload-self-owned}.
\end{proof}

Proposition~\ref{proposi-workload-self-owned} has the following implications. In spite of the relation of $\beta_{0}$ and $\beta$, the workload $z_{i}^{o}$ processed by spot instances is linearly proportional to the parallelism bound $\delta_{i}$ and the additional available time $x_{i}$ for executing task $i$ until some threshold, after which the workload $z_{i}^{o}$ keeps constant and stops increasing with $x_{i}$. This is the same as the case in Section~\ref{sec.case-spot} where only spot and on-demand instances are utilized. Thus, in the case with self-owned instances ({\em i.e.,} $r>0$), we can still apply Algorithm~\ref{Algo-deadline-allocation-basic} to determine the optimal allocation of deadlines: the specific way is presented in lines 1-5 of Algorithm~\ref{Algo-instance-allocation-process}.

\subsection{Summarizing Deadline and Instance Allocation}

In this subsection, we summarize the process of allocating instances to a chain of tasks.

As the time horizon expands, we check whether specific events are triggered at every moment $t$ and take corresponding allocation actions, which are presented in Algorithm~\ref{Algo-instance-allocation-process}. Generally, when a job $j$ arrives, we first determine its deadline allocation. For all $i\in$ $[1, l]$, its $i$-th task can be executed when its preceding tasks have been finished if any. The $l$ tasks are executed one by one. In particular, when $t=a_{j}$, job $j$ arrives and we first determine the allocation of deadlines $\varsigma_{1}, \varsigma_{2}, \cdots, \varsigma_{l}$ to its $l$ tasks (lines 1-5): when only on-demand and spot instances are utilized, execute lines 1-3 since $r=0$; otherwise, execute lines 1-5 since $r>0$.

Recall that $\varsigma_{0}=a_{j}$.
For all $i\in [1, l]$, when $t$ $=$ $\varsigma_{i-1}$, it means that either job $j$ just arrives if $i=1$ or the ($i-1$)-th task has been finished if $i\geq 2$; then, the execution of the $i$-th task begins and we determine the instance allocation to task $i$ (lines 6-15). In the case that there are self-owned instances ({\em i.e.}, $r>0$), when $t=\varsigma_{i-1}$, task $i$ is first allocated $r_{i}$ self-owned instances in $[\varsigma_{i-1}, \varsigma_{i}]$, where $r_{i}$ is given in (\ref{policy-self-owned}) (lines 6-8); otherwise, $r_{i}=0$ (lines 9-10). If $r_{i}>0$, task $i$ can be viewed as a new task with reduced parallelism bound and task size that will only be processed by spot and on-demand instances. Except the possible workload processed by self-owned instances, the remaining workload of task $i$ to be processed at time $t$ is denoted by $\tilde{z}_{i}(t)$. While task $i$ is being executed at time $t\in [\varsigma_{i-1}, \varsigma_{i}]$, if $\tilde{z}_{i}(t)=0$, no actions are taken to request spot and on-demand instances since the current allocation of instances is enough to finish task $i$; otherwise,
we have
\begin{itemize}
\item if there is flexibility for task $i$ to utilize spot instances at time $t$ by Definition~\ref{def-flexibility}, request to utilize $\delta_{i}-r_{i}$ spot instances (lines 12-13).
\item otherwise, there is no such flexibility and $t$ is the turning point of task $i$; by Definition~\ref{def-allocation-process}, stop requesting spot instances and turn to utilize $\delta_{i}-r_{i}$ on-demand instances in $[t, \varsigma_{i}]$ (lines 14-15).
\end{itemize}

\begin{algorithm}[t]
\SetKwInOut{Input}{Input}
\SetKwInOut{Output}{Output}

\tcc{\footnotesize{Check the possible events at time $t$}}

   \If{there exists some job $j$ such that $t=a_{j}$}{

    \tcp{\footnotesize{Allocate deadlines to job $j$}}

       \If{$r=0$ or ($r>0$ and $\beta < \beta_{0}$)}{
           Call Dealloc($\beta$), presented in Algorithm~\ref{Algo-deadline-allocation-basic}\;
       }
       \If{$r>0$ and $\beta_{0}\leq \beta$}{
           Call Dealloc($\beta_{0}$), presented in Algorithm~\ref{Algo-deadline-allocation-basic}\;
       }
   }

    \If{$t=\varsigma_{i-1}$}{

       \tcp{\footnotesize{Allocate self-owned instances to task $i$, if any}}

       \If{$r>0$}{
           Allocate $r_{i}$ self-owned instances to task $i$ where $r_{i}$ is given in (\ref{policy-self-owned})\;
       }

       \If{$r=0$}{
           No self-owned instances are allocated to task $i$ where $r_{i}=0$\;
       }
    }

    \If{$t\in [\varsigma_{i}, \varsigma_{i+1}]$ and $\tilde{z}_{i}(t)>0$}{

       \tcp{\footnotesize{Determine the allocation action of spot and on-demand instances}}

        \If{by Definition~\ref{def-flexibility}, there is flexibility for task $i$ to utilize spot instances at time $t$}{

            Bid a price for $\delta_{i}-r_{i}$ spot instances\;
        }
        \Else{
           \tcp{\footnotesize{$t$ is the turning point of task $i$}}
            Request $\delta_{i}-r_{i}$ on-demand instances in the time window $[t, \varsigma_{i}]$\;
        }
    }

\caption{Deadline and Instance Allocation}\label{Algo-instance-allocation-process}
\end{algorithm}



\section{Online Learning for Generalized Case}
\label{sec.generalized-online-learning}

{In the last section, we propose a series of parametric policies for allocating instances to a chain of tasks, which are the core technical contribution of this paper.} Supported by {two existing techniques directly from \cite{Nagarajan13a,Jain14}}, we can further obtain an integrated framework to process general DAG jobs, which is of great interest in practice.
In this section, we introduce the two techniques briefly, although they are not the main contribution of this paper. Their formal description is given in {Appendix~\ref{append-description}}.



\vspace{0.35em}\noindent\textbf{Job Transformation.} The technique of Nagarajan {\em et al.} \cite{Nagarajan13a} is used to transform a general DAG job $j$ to a virtual job $j^{\prime}$ with a chain precedence constraint, also called a pseudo-job. Any feasible schedule of the pseudo-job $j^{\prime}$ is also a feasible schedule of the DAG job $j$, with their parallelism, precedence and deadline constraints respected. While transforming $j$ to $j^{\prime}$, the high-level idea is as follows. Consider a virtual schedule of $j$, also called a pseudo-schedule: each task $i$ of $j$ is allocated $\delta_{i}$ instances and executed as early as possible. Each pseudo-task of $j^{\prime}$ consists of parts of the tasks of $j$ that are executed in the same time interval. There are multiple time intervals between the starting and completion times of the pseudo-schedule. These intervals correspond to multiple pseudo-tasks that form a pseudo-job $j^{\prime}$ with a chain precedence constraint.

\vspace{0.35em}\noindent\textbf{Learning the Optimal Parameters.} The online learning algorithm (TOLA) of Menache {\em et al.} \cite{Jain14} is adapted to learn the most cost-effective parametric policy.

Each job is associated with a particular parametric policy that is defined by a tuple of parameters $\{\beta, \beta_{0}, b\}$. 
The parameter $\beta$ represents the availability of spot service while $\beta_{0}$ indicates the sufficiency of self-owned instances. 
When a job $j$ arrives, $\beta$ and $\beta_{0}$ are used to determine the deadline allocation via the lines 1-5 of Algorithm~\ref{Algo-instance-allocation-process}. The value of $\beta$ may only depend on the system dynamics, independent of the behavior of an individual user; this is the case of Google Cloud. Besides the system dynamics, it may also relate to the bid price $b$ of a user; this is the case of Amazon EC2 and Microsoft Azure. Then, a user needs to bid a price $b$ to request spot instances; its jobs fail to get instances when either $b$ is lower than the spot price at a moment or the system reclaims the allocated instances. In this case, we need to learn the best bid price $b$ against the spot price dynamics. In the case of Google Cloud, no bid is required and we simply set $b$ to a null value.

There is a set $\mathcal{P}$ of $n$ tuples $\{\beta, \beta_{0}, b\}$, each representing one policy. The high-level idea of TOLA is as follows. This is an initial probability distribution over the $n$ policies. Whenever a job $j$ arrives at time $t$, a policy is randomly chosen from $\mathcal{P}$ according to the distribution and it determines the actual allocation of instances to the job $j$ and the actual cost of completing $j$. On the other hand, given an arbitrary policy, the cost of completing an arbitrary job $j^{\prime}$ depends on the fixed on-demand price and the variable spot prices in $[a_{j^{\prime}}, d_{j^{\prime}}]$. At time $t$, for the past jobs whose deadlines are no larger than $t$, we can derive their costs under each policy of $\mathcal{P}$ since we know the spot prices in $[0, t]$. We can choose one of such jobs unexamined so far, and examine its cost under each policy; then the distribution is updated at time $t$ such that the lower-cost (higher-cost) polices of this job are re-assigned the enlarged ({\em resp.} reduced) probabilities.

As the time horizon expands, the probability distribution is updated over and over and the most cost-effective policies of $\mathcal{P}$ will be identified gradually, {\em i.e.}, the ones with the highest probabilities. In the meantime, as more and more jobs are processed, the actual cost of completing all jobs will be close to the cost of completing all jobs under the best policy of $\mathcal{P}$.

\section{Evaluation}
\label{sec.evaluation}

The main aim of our evaluations is to show the effectiveness of the proposed policies of this paper.

\subsection{Simulation Setups}
\label{sec.parameters}

{In alignment to best practices in prior art} \cite{Jain14,Nagarajan13a,Wu20a}, jobs are generated as follows. The on-demand price $p$ is normalized to be 1. The job arrival follows a poisson process with a mean of 4. The number $l$ of tasks in a job is randomly set to $7$ or $49$. The order of generating tasks is also the topological order of tasks in the graph. For any two tasks $i_{1}$ and $i_{2}$, a precedence constraint is associated with a probability 0.5. To ensure connectivity, for all $i\in [1, l-1]$, a task $i$ without successors is randomly connected to one of the latter tasks $i+1, \cdots, l$, as its successor; for all $i\in [2, l]$, a task $i$ without predecessors is randomly connected to one of its former tasks $1, \cdots, i-1$ as its predecesor. The parallelism bound of a task is randomly set to 8 and 64. The minimum execution time $e_{i}$ of every task $i$ follows a bounded Pareto distribution \cite{Chen11a} with a shape parameter $\epsilon = \frac{7}{8}$, a scale parameter $\sigma=\frac{7}{32}$ and a location parameter $\mu=\frac{1}{4}$; the maximum and minimum values of $x$ are set to 2 and 10. The task size $z_{i}$ is $e_{i}\cdot\delta_{i}$.

For each DAG job $j$, we compute its critical path and denote its length by $e_{j}^{c}$, which is the minimum execution time needed to finish $j$ \cite{Nagarajan13a}. The job's relative deadline $d_{j}-a_{j}$ is set to $x\cdot e_{j}^{c}$, where $x$ is uniformly distributed over $[1, x_{0}]$. $x$ represents jobs' flexibility and determines their capability to utilize spot instances; it is a main factor that determines the performance. In this paper, we consider four types of jobs with different levels of time flexibility, and {\em the 1st, 2nd, 3rd and 4th types of jobs} respectively have $x_{0} = 1.5, 2, 2.5, 3$. Each DAG job is transformed into a simpler job with chain-like precedence constraints, after which various policies are applied to the simplified job for processing. We can use an exponential distribution to model spot prices \cite{Zheng15}. Specifically, each unit of time is divided into 12 equal time slots, and spot prices are updated per slot; their values can follow a bounded exponential distribution where its mean is set to 0.13; the upper and lower bounds are set to 1 and 0.12.

\vspace{0.3em}\noindent\textbf{Proposed Policies.} The parametric policy $\{\beta_{0}, \beta, b\}$ is described in Section~\ref{sec.generalized-online-learning}. $\beta_{0}$, $\beta$ and $b$ are chosen respectively from $\mathcal{C}_{1}=$ $\left\{\frac{2}{12}, \frac{4}{14}, \frac{6}{16}, \frac{8}{18}, \frac{1}{2}, 0.6, 0.7\right\}$, $\mathcal{C}_{2} = \left\{1, \frac{1}{1.3}, \frac{1}{1.6}, \frac{1}{1.9}, \frac{1}{2.2}\right\}$, and $\mathcal{B}$ $= \{ 0.18, 0.21, 0.24, 0.27, 0.3\}$. When only spot and on-demand instances are considered, the set of policies is set to $$\boldsymbol{\mathcal{P}}=\{(\beta, b) \mid \beta\in \mathcal{C}_{2}, b\in\mathcal{B} \}.$$ When there are also self-owned instances, the set of policies is set to $$\boldsymbol{\mathcal{P}}=\{(\beta, b, \beta_{0}) \mid \beta_{0}\in\mathcal{C}_{1}, \beta\in \mathcal{C}_{2}, b\in\mathcal{B} \}.$$

\vspace{0.3em}\noindent\textbf{Benchmark Policies.} The benchmark policy is used as a baseline to measure the performance of the proposed policy. Our analysis in Proposition~\ref{proposi-spot} formalizes that an intuitive policy can achieve the expected optimal utilization of spot instances. For comparison, the benchmark policies include (\rmnum{1}) the naive policy for allocating the time windows in which tasks are executed and (\rmnum{2}) the naive policy for self-owned instances. We evaluate {\em two possible naive policies for time window allocation}, where the first can only be applied to spot and on-demand instances and the other will also be applied to self-owned instances:
\begin{description}
\item [\textbf{Greedy}] As the time horizon expands, a job simply bids for $\delta_{i}$ spot instances for each of its tasks until the length of the critical path for processing the remaining workload of tasks is no less than the remaining time window size; afterwards, we simply use $\delta_{i}$ on-demand instances for processing the remaining workload of each task $i$.

\vspace{0.16em}\item [\textbf{Even}] Upon arrival of a job, we specify a series of consecutive time windows in which its tasks are executed and finished. Each task $i$ has a time window size $\varsigma_{i}=e_{i}+x_{i}$. The remaining time $\omega=d_{j}-a_{j}-\sum_{i=1}^{l}{e_{i}}$ is evenly allocated among the $l$ tasks, and we set $x_{i}$ to $\omega/l$.
\end{description}
{\em The naive policy for self-owned instances} would be allocating as many self-owned instances as possible to each task in a first-come-first-served discipline, taking into account the number of self-owned instances available. Specifically, upon arrival of a job, if the time windows of tasks are specified, we allocate as many self-owned instances as possible to each task $i$ within its parallelism bound, {\em i.e.}, $$r_{i}=\min\{N(\varsigma_{i-1}, \varsigma_{i}),\, \delta_{i}\}.$$ The set of benchmark policies are parameterized and defined as $$\boldsymbol{\mathcal{P}^{\prime}} = \{ b \mid b\in\mathcal{B} \}.$$

\vspace{0.3em}\noindent\textbf{Performance Metric.} The objective of this paper is minimizing the cost of finishing a set of jobs $\mathcal{J}$ that arrive over time. Each job $j$ is processed under a proposed or benchmark policy, indexed by $\pi$. There are three types of jobs to be evaluated. Let $Z_{j}$ denote the total workload of job $j$ that consists of $l$ tasks, {\em i.e.}, $Z_{j}=\sum_{i=1}^{l}{z_{i}}$. {Let $c_{j}(\pi)$ denote the cost of completing $j$ under the policy $\pi$.} When there are $x_{1}$ self-owned instances and the $x_{2}$-th type of jobs are processed, {\em the average unit cost of processing jobs under a policy $\pi$}, denoted by $\alpha_{x_{1}, x_{2}}(\pi)$, is defined as the ratio of the total cost of utilizing various instances to the processed workload of jobs: $$\alpha_{x_{1}, x_{2}}(\pi) = \sum_{j\in\mathcal{J}}{c_{j}(\pi)}/\sum_{j\in\mathcal{J}}{Z_{j}}.$$

When a fixed policy $\pi$ is applied to all jobs, we use $\alpha_{x_{1}, x_{2}}$ ({\em resp.} $\alpha_{x_{1}, x_{2}}^{\prime}$) to denote the minimum of the average unit costs of our proposed policies ({\em resp.} the benchmark policies):
\begin{align*}
\alpha_{x_{1}, x_{2}}=\min\limits_{\pi\in\mathcal{P}}{\alpha_{x_{1}, x_{2}}(\pi)} \,\text{ and }\, \alpha_{x_{1}, x_{2}}^{\prime}=\min\limits_{\pi\in\mathcal{P}^{\prime}}{\alpha_{x_{1}, x_{2}}(\pi)}.
\end{align*}
To measure the effectiveness of our proposed policies over the benchmark policies, we define a metric, called {\em cost improvement}, as follows:
\begin{center}
$\rho_{x_{1}, x_{2}} = 1 - \frac{\alpha_{x_{1}, x_{2}}}{\alpha_{x_{1}, x_{2}}^{\prime}}$;
\end{center}
$\rho_{x_{1}, x_{2}}$ represents how much cost is saved by using our proposed policies, compared with the benchmark policies. For example, when $\rho_{x_{1}, x_{2}}=0.5$, the cost of our proposed policies is only half the cost of the benchmark policies.

Furthermore, in this paper, the policies of a set are associated with a probability distribution on which we base the selection of a policy for each arriving job. The online learning algorithm TOLA ({\em i.e.}, Algorithm~\ref{Regret} in Appendix~\ref{sec.learning}) is run to update the distribution, finally identifying the policy that generates the lowest cost. When TOLA is applied, we use $\overline{\alpha}_{x_{1}, x_{2}}(\mathcal{P})$ ({\em resp.} $\overline{\alpha}_{x_{1}, x_{2}}(\mathcal{P}^{\prime})$) to denote the average unit cost of processing all jobs if the set of policies is $\mathcal{P}$ ({\em resp.} $\mathcal{P}^{\prime}$), and the cost improvement is defined as follows:
\begin{center}
$\overline{\rho}_{x_{1}, x_{2}} = 1 - \frac{\overline{\alpha}_{x_{1}, x_{2}}(\mathcal{P})}{\overline{\alpha}_{x_{1}, x_{2}}(\mathcal{P}^{\prime})}$.
\end{center}
$\overline{\rho}_{x_{1}, x_{2}}$ represents the cost saving when online learning is applied.

\subsection{Results}
\label{sec.evaluation-results}

Our simulations are run over about 10000 jobs. We will show the cost improvement of our proposed policies over the benchmark policies.

\vspace{0.25em}\noindent\textbf{Experiment 1.} We evaluate the effectiveness of the proposed deadline allocation algorithm ({\em i.e.}, Algorithm~\ref{Algo-deadline-allocation-basic}) in the case that a user {does not have any} self-owned instances ({\em i.e.}, $x_{1}=0$) and only utilizes spot and on-demand instances. This algorithm is compared with the greedy and even policies in Section~\ref{sec.parameters}. The corresponding results are listed in Table~\ref{table-spot-optimal}.
The cost improvement of our algorithm is significant and ranges from 15.23\% to 27.10\%. The improvement is especially strong when the population of jobs has a tight time flexibility to be finished, {\em e.g.}, the cost improvement can be up to 27.10\%. Since our proposed policy is expected to be optimal, we can see in all cases that the cost of our policy is a lower bound of the cost of the other policies.

\begin{table}[t]
	\centering
		\caption{Cost Improvement for Spot and On-Demand Instances}
	\begin{threeparttable}[b]

		\begin{tabular}{|C{1.3cm}|C{1.2cm}| C{1.2cm} | C{1.2cm} | C{1.2cm} |}   
			
			\hline
			           & $\rho_{0, 1}$   & $\rho_{0, 2}$   & $\rho_{0,3}$  & $\rho_{0,4}$  \\ \hline	
     \textbf{Greedy}   &   27.10\%       &     20.90\%     &    16.53\%    &  15.23\%   \\ \hline	
     \textbf{Even}     &   25.61\%       &     22.20\%     &    18.03\%    &  16.39\%   \\ \hline	

		\end{tabular}
	\end{threeparttable}
		\label{table-spot-optimal}
\end{table}

\vspace{0.25em}\noindent\textbf{Experiment 2.} We consider the case that a user also has some self-owned instances. In our proposed framework ({\em i.e.,} Algorithm~\ref{Algo-instance-allocation-process}), there are policies for allocating deadlines and self-owned instances. We evaluate the overall effectiveness of this framework. The benchmark policies for comparison include the even policy for allocating deadlines and the naive policy for allocating self-owned instances. The corresponding results are listed in Table~\ref{table-self-owned}. The cost improvement is significant and ranges from 37.22\% to 62.73\%. {As a user has more self-owned instances, less spot and on-demand instances will be consumed to complete all the jobs. The more self-owned instances a user has, the larger their effect on the cost.} With our proposed policies, the cost improvement increases as the number of self-owned instances increases from 300 to 1200.


\begin{table}[t]
	\centering
		\caption{Overall Cost Improvement with Self-Owned Instances}
	\begin{threeparttable}[b]
		\begin{tabular}{|C{0.8cm}| C{1.06cm} |  C{1.06cm} |  C{1.06cm} |  C{1.06cm} |}
			\hline
\diagbox{$x_{1}$}{$\rho$}{$x_{2}$}  &     $1$     &   $2$      &    $3$      & $4$         \\ \hline
                           $300$    &   37.22\%   &  41.28\%   &   39.57\%   & 37.26\%     \\ \hline	
                           $600$    &   43.60\%   &  51.43\%   &   50.05\%   & 45.79\%     \\ \hline
                           $900$    &   50.57\%   &  58.80\%   &   55.06\%   & 50.81\%     \\ \hline
                           $1200$   &   57.95\%   &  62.73\%   &   58.57\%   & 55.24\%     \\ \hline
		\end{tabular}
	\end{threeparttable}
	\label{table-self-owned}
\end{table}

\begin{table}[t]
	\centering
		\caption{Cost Improvement for Self-Owned Instances}
	\begin{threeparttable}[b]
		\begin{tabular}{|C{0.8cm}| C{1.06cm} |  C{1.06cm} |  C{1.06cm} |  C{1.06cm} |}
			\hline
\diagbox{$x_{1}$}{$\rho$}{$x_{2}$}  &     $1$     &   $2$      &    $3$      & $4$         \\ \hline
                           $300$    &   13.16\%   &  18.30\%   &   20.14\%   & 20.51\%     \\ \hline	
                           $600$    &   21.25\%   &  31.97\%   &   33.74\%   & 31.00\%     \\ \hline
                           $900$    &   30.64\%   &  42.14\%   &   40.13\%   & 37.04\%     \\ \hline
                           $1200$   &   40.68\%   &  47.37\%   &   44.60\%   & 42.50\%     \\ \hline
		\end{tabular}
	\end{threeparttable}
	\label{table-self-owned-1}
\end{table}

\begin{table}[!ht]
	\centering
		\caption{Utilization Ratio $\mu_{x_{1}, x_{2}}$ for Self-Owned Instances}
	\begin{threeparttable}[b]
		\begin{tabular}{|C{0.8cm}| C{1.06cm} |  C{1.06cm} |  C{1.06cm} |  C{1.06cm} |}
			\hline
\diagbox{$x_{1}$}{$\mu$}{$x_{2}$}  &     $1$     &   $2$      &    $3$      & $4$         \\ \hline
                           $300$    &   97.01\%   &  93.02\%   &   88.46\%   & 81.68\%     \\ \hline	
                           $600$    &   96.40\%   &  90.58\%   &   79.95\%   & 73.98\%     \\ \hline
                           $900$    &   95.81\%   &  82.13\%   &   73.08\%   & 76.64\%     \\ \hline
                           $1200$   &   94.58\%   &  78.59\%   &   79.28\%   & 74.00\%     \\ \hline
		\end{tabular}
	\end{threeparttable}
	\label{table-utilization-comparison}
\end{table}

\vspace{0.25em}\noindent\textbf{Experiment 3.} We still consider the case with all the three types of instances, like Experiment 2. We evaluate the effectiveness of the proposed policy (\ref{policy-self-owned}) for allocating self-owned instances, by comparison with the benchmark policy for self-owned instances. While evaluating these two policies, the same deadline allocation algorithm ({\em i.e.}, lines 1-5 of Algorithm~\ref{Algo-instance-allocation-process}) is used. The corresponding results are listed in Table~\ref{table-self-owned-1}. The cost improvement is significant and ranges from 13.16\% to 47.37\%. Given the job type $x_{2}$, the cost improvement increases with the amount of self-owned instances that a user possesses.

On the other hand, we are also interested in the utilization of self-owned instances and show the ratio of the utilization of our proposed policy to the utilization of the benchmark policy. We use $\mu_{x_{1}, x_{2}}$ to denote this ratio under the $x_{2}$-th type of jobs when there are $x_{1}$ self-owned instances. The corresponding results are listed in Table~\ref{table-utilization-comparison}. For example, when $x_{1}=900$ and $x_{2}=1$, the utilization ratio $\mu_{x_{1}, x_{2}}$ is 0.9581. Overall, the benchmark policy can achieve a higher utilization of self-owned instances. It allocates as many self-owned instances as possible to the tasks of each job. However, jobs are differentiated by their capability to utilize spot instances. We should allocate more self-owned instances to the jobs that have poor capability to utilize spot instances, which can effectively reduce the unnecessary consumption of costly on-demand instances. This leads to that, although our proposed policy achieves a lower utilization, it can still achieve a lower cost than the benchmark policy. For example, when $x_{1}=900$ and $x_{2}=1$, the cost improvement $\rho_{x_{1}, x_{2}}$ is 30.64\%.

\vspace{0.25em}\noindent\textbf{Experiment 4.} Finally, we show the performance when the online learning algorithm TOLA, is applied. When only spot and on-demand instances are considered, the experimental setting here is the same as Experiment 1. When self-owned instances are also considered, the experimental setting is the same as Experiment 2. We list in Table~\ref{table-regret} the values of $\rho_{x_{1}, x_{2}}$ in the case that the job type $x_{2}$ is $2$ and the number $x_{1}$ of self-owned instances is 0, 300, 600, 900 and 1200 respectively. The results still show a significant cost improvement, ranging from 24.87\% to 59.05\%.

\begin{table}[t]
	\centering
		\caption{Cost Improvement under Online Learning}
	\begin{threeparttable}[b]
		\begin{tabular}{| C{1.2cm} | C{1.2cm} |  C{1.2cm} |  C{1.2cm} |  C{1.2cm} |}
			\hline
		$\overline{\rho}_{0,2}$    &  $\overline{\rho}_{300,2}$   & $\overline{\rho}_{600,2}$   & $\overline{\rho}_{900, 2}$   & $\overline{\rho}_{1200, 2}$  \\ \hline
     24.87\%     &   36.91\%    &  47.26\%  &  54.71\%  &  59.05\%   \\ \hline	

		\end{tabular}
	\end{threeparttable}
		\label{table-regret}
\end{table}

\section{Conclusion}
\label{sec.conclusion}

The formation of cost-effectively using IaaS clouds opens the door for users to participate in cloud ecosystems. We consider DAG jobs that are an important extension to the independent tasks considered in the previous works. A job has a specific deadline and consists of multiple tasks with precedence constraints. Driven by the goal of maximizing the utilization of self-owned and spot instances, we identify that a key question is allocating deadlines to tasks. Thus, given the policies for allocating instances, we qualitatively characterize the capability that each task has to utilize spot instances in a predefined time window. Based on this, we formulate the deadline allocation problem as an integer program and derive {in a computationally efficient fashion the optimal solution}. These policies and algorithms are parametric and thus adaptive. Facing the dynamic of cloud market, we can leverage online learning to infer their optimal values. Several intuitive heuristics are used as baselines to validate the cost improvement brought by the proposed solutions. The cost improvement is up to 24.87\% when spot and on-demand instances are considered and up to 59.05\% when self-owned instances are considered.







\appendices

\section{}

\subsection{Proof of Proposition~\ref{proposi-spot}}
\label{sec.spot}

By Definition~\ref{def-allocation-process}, if $\hat{\varsigma}_{i} > e_{i}$, we have that either there will be a turning point $\varsigma_{i}^{c}$ larger than $\varsigma_{i-1}$ or such a turning point does not exist, since (\ref{equa-flexibility}) holds at time $t=\varsigma_{i-1}$. Thus, from time $\varsigma_{i-1}$ on, $o_{i}$ on-demand instances and $s_{i}$ spot instances are requested. The duration of utilizing spot instances is denoted by $\tau_{i}$ where $\tau_{i}\in (0, \hat{\varsigma}_{i}]$. The amount of workload processed by spot instances is $s_{i}\cdot\tau_{i}\cdot \beta$ and it achieves the maximum possible value when $s_{i}=\delta_{i}$ and $\tau_{i}=\hat{\varsigma}_{i}$. Thus, a necessary condition under which task $i$ can be finished by simply utilizing spot instances is $\delta_{i}\cdot \hat{\varsigma}_{i}\cdot \beta\geq z_{i}$, {\em i.e.}, the condition (\ref{condition-1}). Under this condition, the corresponding optimal strategy is to request $\delta_{i}$ spot instances where $s_{i}=\delta_{i}$ and $o_{i}=0$ and it is expected that the turning point does not exist.

If $\hat{\varsigma}_{i} \in \left(e_{i}, \frac{e_{i}}{\beta}\right)$, there will be a turning point $\varsigma_{i}^{c}$ larger than $\varsigma_{i-1}$ where $\tau_{i}=\varsigma_{i}^{c}-\varsigma_{i-1}$ and the instance allocation process has two phases by Definition~\ref{def-allocation-process}; then, we have
\begin{align}\label{equa-workload-second-phase}
\tau_{i}\cdot s_{i}\cdot \beta + o_{i}\cdot \tau_{i} + \delta_{i}\cdot (\hat{\varsigma}_{i}-\tau_{i})=z_{i}.
\end{align}
By (\ref{equa-parallelism-no-self-owned}) and (\ref{equa-workload-second-phase}), we can derive that the workload processed by spot instances is
\begin{align}\label{equa-spot-workload}
\tau_{i}\cdot s_{i}\cdot \beta=\frac{\beta}{1-\beta}\cdot (\delta_{i}\cdot \hat{\varsigma}_{i}-z_{i}).
\end{align}
The right side of Equation (\ref{equa-spot-workload}) is independent of the values of $s_{i}$ and $\tau_{i}$; here, $\delta_{i}$, $\hat{\varsigma}_{i}$, and $z_{i}$ are known. Thus, if $\hat{\varsigma}_{i} \in \left(e_{i}, \frac{e_{i}}{\beta}\right)$, the optimal strategy can be to request $\delta_{i}$ spot instances in the first phase and $\delta_{i}$ on-demand instances in the second phase.

Finally, if $\hat{\varsigma}_{i} = e_{i}$, we have that the turning point exists and $\varsigma_{i}^{c}=\varsigma_{i-1}$. By Definition~\ref{def-allocation-process}, we have $o_{i}=\delta_{i}$ and $s_{i}=0$.

\subsection{Proof of Proposition~\ref{proposi-spot-workload}}
\label{sec.spot-workload}

With different time window size, it is expected that task $i$ is finished by different composition of instances. We have the following observations by Proposition~\ref{proposi-spot}. If $\hat{\varsigma}_{i}=e_{i}$, no workload is processed by spot instances and $z_{i}^{o}=0$. If $\hat{\varsigma}_{i}\in \left(e_{i}, \frac{e_{i}}{\beta} \right)$, the workload processed by spot instances is defined in (\ref{equa-spot-workload}) and we have by (\ref{equa-minimum-e-time}) and (\ref{equa-window-structure}) that $z_{i}^{o}=\frac{\beta}{1-\beta}\cdot \delta_{i}\cdot x_{i}$. If $\hat{\varsigma}_{i} \geq \frac{e_{i}}{\beta}$, it is expected that task $i$ is completed by utilizing spot instances alone and the workload processed by spot instances is $z_{i}$. Specifically, when $\hat{\varsigma}_{i} = \frac{e_{i}}{\beta}$, we have by (\ref{equa-minimum-e-time}) and (\ref{equa-window-structure}) that $z_{i}^{o}=z_{i}=\frac{\beta}{1-\beta}\cdot \delta_{i}\cdot x_{i}$. The proposition thus holds.

\subsection{Proof of Proposition~\ref{proposi-optimal-deadline-allocation}}
\label{sec.optimal-deadline-allocation}

We will prove by contradiction that any solution different from the one of Algorithm~\ref{Algo-deadline-allocation-basic} cannot lead to that more workload is processed by spot instances. Suppose that in an optimal solution there exists a task $i_{k^{\prime}}$, where $k^{\prime}\leq l^{\ast}$, whose time window size is smaller than $\hat{\varsigma}_{i_{k^{\prime}}}^{\ast}$. With abuse of notation, let $k^{\prime}$ be the maximum such integer in $[1, l^{\ast}]$. We can find some latter tasks $i$, where $i> l^{\ast}$, whose time window sizes are larger than $\hat{\varsigma}_{i}^{\ast}$. Then, we reduce the allocated time window sizes of these tasks by $\hat{\varsigma}_{i_{k^{\prime}}}^{\ast}-\hat{\varsigma}_{i_{k^{\prime}}}$ while keeping their sizes $\geq e_{i}$; correspondingly, we increase the time window size $\hat{\varsigma}_{i_{k^{\prime}}}$ to $\hat{\varsigma}_{i_{k^{\prime}}}^{\ast}$. By Proposition~\ref{proposi-spot-workload}, this will lead to that the workload processed by spot instances is the same as or larger than the workload before the re-allocation of time window sizes, since the parallelism bound of task $i_{k^{\prime}}$ is no smaller than the latter ones. Thus, the solution produced by Algorithm~\ref{Algo-deadline-allocation-basic} can lead to that the most workload is processed by spot instances.

\subsection{Proof of Proposition~\ref{property-self-owned}}
\label{sec.property-self-owned}

We also have $f(x)=\max  \left\{  \delta_{i} - \frac{\delta_{i}\cdot \hat{\varsigma}_{i}-z_{i}}{\hat{\varsigma}_{i}\cdot (1-x)},\,\, 0 \right\}$. Since $\delta_{i}\cdot \hat{\varsigma}_{i}-z_{i}\geq 0$, the proposition's second point holds. Next, we prove the first point. Suppose task $i$ is allocated $r_{i}$ self-owned instances. In order to ensure that the remaining workload $z_{i}-r_{i}\cdot \hat{\varsigma}_{i}$ of task $i$ can be finished by totally utilizing spot instances, we have $\beta\cdot (\delta_{i}-r_{i}) \cdot \hat{\varsigma}_{i} \geq z_{i}-r_{i}\cdot \hat{\varsigma}_{i}$. We further have
\begin{align*}
r_{i} \geq \frac{z_{i}-\delta_{i}\cdot \hat{\varsigma}_{i}\cdot \beta}{\hat{\varsigma}_{i}\cdot (1-\beta)}.
\end{align*}
Since $r_{i}\geq 0$, the proposition's first point holds.

\subsection{Proof of Proposition~\ref{proposi-workload-self-owned}}
\label{sec.workload-self-owned}

{\em First}, we study the case that $\beta_{0}\leq \beta$ and task $i$ is allocated $r_{i} = f(\beta_{0})$ self-owned instances; then, task $i$ is expected to be finished by utilizing self-owned and spot instances alone by Proposition~\ref{property-self-owned}, where $f(\beta_{0})\geq f(\beta)$. If $\hat{\varsigma}_{i} \geq \frac{e_{i}}{\beta_{0}}$, we have $r_{i}=0$ since $f(\beta_{0})=0$ by (\ref{equa-f-x}); then, $z_{i}^{o} = z_{i}$. If $\hat{\varsigma}_{i} \in \left[e_{i},  \frac{e_{i}}{\beta_{0}}\right]$, we have $r_{i}=$ $f(\beta_{0})=\frac{z_{i}-\delta_{i}\cdot \hat{\varsigma}_{i}\cdot \beta_{0}}{\hat{\varsigma}_{i}\cdot (1-\beta_{0})} \geq 0$; then, the workload processed by spot instances is
\begin{align}
z_{i}^{o} = z_{i} - r_{i}\cdot \hat{\varsigma}_{i} = \frac{\beta_{0}}{1-\beta_{0}}\cdot \delta_{i}\cdot x_{i}.
\end{align}
Thus, the proposition's first point holds.

{\em Second}, we study the case that $\beta_{0}> \beta$. If $\hat{\varsigma}_{i} \geq \frac{e_{i}}{\beta}$, we have by Proposition \ref{proposi-spot} that task $i$ is expected to be finished by only requesting spot instances in $[\varsigma_{i-1}, \varsigma_{i}]$, without self-owned and on-demand instances and that $f(\beta) = 0$. Thus, $f(\beta_{0}) = 0$ since $0\leq$ $f(\beta_{0})$ $\leq f(\beta)$ by Proposition~\ref{property-self-owned}. As a result, we have $z_{i}^{o} = z_{i}$. If $\hat{\varsigma}_{i}$ $\in \left[e_{i},  \frac{e_{i}}{\beta}\right)$, we have the following observation where $f(\beta) > 0$ by (\ref{equa-f-x}). Task $i$ is allocated $f(\beta_{0})$ self-owned instances where $0\leq f(\beta_{0})< f(\beta)$; then, task $i$ is expected to utilize some on-demand instances to be finished by its deadline. Suppose that the duration for which task $i$ can utilize spot instances is $\tau_{i}$ where $\tau_{i}\in [0, \hat{\varsigma}_{i})$. The workload of task $i$ is processed respectively by $r_{i}$ self-owned instances for a duration $\hat{\varsigma}_{i}$, $\delta_{i}-r_{i}$ spot instances for a duration $\tau_{i}$, and $\delta_{i}-r_{i}$ on-demand instances for a duration $\hat{\varsigma}_{i}-\tau_{i}$, and have
\begin{align*}
z_{i} - r_{i}\cdot \hat{\varsigma}_{i} = \beta\cdot\tau_{i}\cdot(\delta_{i}-r_{i}) + (\hat{\varsigma}_{i}-\tau_{i})\cdot(\delta_{i}-r_{i}).
\end{align*}
Further, we can derive $\tau_{i}=\frac{\delta_{i}\cdot\hat{\varsigma}_{i}-z_{i}}{(1-\beta)\cdot(\delta_{i}-r_{i})}$. The workload processed by spot instances is
\begin{align}
z_{i}^{o} = (\delta_{i}-r_{i})\cdot\tau_{i}\cdot\beta = \frac{\beta}{1-\beta}\cdot\delta_{i}\cdot x_{i}.
\end{align}
Here, when $\hat{\varsigma}_{i} = e_{i}+x_{i} = \frac{e_{i}}{\beta}$, we have $z_{i}^{o}=z_{i}=\frac{\beta}{1-\beta}\cdot\delta_{i}\cdot x_{i}$. This completes the proof of the second point.

\section{}
\label{append-description}

\subsection{Job Transformation}


In this subsection, we describe the technique of Nagarajan {\em et al.} \cite{Nagarajan13a} formally. Consider a DAG job $j$ of $l$ tasks where each task $i$ has a size $z_{i}$ and a parallelism bound $\delta_{i}$. We first give a pseudo-schedule on which we base the construction of the corresponding virtual pseudo-job:
\begin{itemize}
\item Allocate each task $i$ the maximum number $\delta_{i}$ of instances and execute $i$ as early as possible; then, we can get the earliest possible start time $q_{i}$ for executing $i$ such that $q_{i}$ $\geq q_{i^{\prime}}+ e_{i^{\prime}}$ for all $i^{\prime} \prec i$.
\item Run task $i$ on $\delta_{i}$ instances in the time window $[q_{i},$ $q_{i} + e_{i}]$.
\end{itemize}
The pseudo-schedule processes the workload of tasks with their precedence constraints respected.

The arrival time of job $j$ is $a_{j}$. Let $T_{j}=\max_{i=1}^{l}\{q_{i} + e_{i}\}$ denote the completion time of $j$ in the pseudo-schedule. For each task $i$ of $j$, the pseudo-schedule defines a specific time window in which the task $i$ runs on $\delta_{i}$ instances. The pseudo-job is constructed as follows:
\begin{itemize}
\item Partition the interval $[a_{j}, T_{j}]$ into the minimum number of sub-intervals $I_{1}$, $I_{2}$, $\cdots$, $I_{l^{\prime}}$ such that if a task runs in a sub-interval, it will run continuously in the entire sub-interval.

\item For $k\in [1, l^{\prime}]$, let $r_{k}$ denotes the total number of instances allocated to all tasks running in $I_{k}$; the pseudo-task is defined to have a parallelism bound $\delta(k)=r_{k}$ and a size $z(k)=r_{k}\cdot |I_{k}|$, which is the total workload processed by the pseudo-schedule during $I_{k}$.

\item Each sub-interval $I_{k}$ corresponds to a pseudo-task; we enforce the chain precedence constraints $1\prec 2\prec \cdots\prec l^{\prime}$ on the $l^{\prime}$ pseudo-tasks.
\end{itemize}
These $l^{\prime}$ pseudo-tasks constitute a pseudo-job $j^{\prime}$.

A general DAG job $j$ is transformed to a pseudo-job $j^{\prime}$ with a chain precedence constraint. The transforming process is denoted by
\begin{align}\label{process-transform}
j^{\prime} \leftarrow \text{transform}(j).
\end{align}
A feasible schedule of the pseudo job $j^{\prime}$ (possibly different from the pseudo-schedule) is also a feasible schedule of the job $j$, which defines the order of executing different parts of the $l$ tasks and still respects the precedence constraints of tasks.

\begin{algorithm}[t]
\SetKwInOut{Input}{Input}
\SetKwInOut{Output}{Output}

\tcc{\footnotesize{transform a DAG job $j^{\prime}$ to a job $j$ with chain-like precedence constraints}}

\If{Job $j^{\prime}$ is not a job with chain-like precedence constraints}{

    $j\leftarrow$ \text{transform}($j^{\prime}$)\;
}\Else{

    $j\leftarrow j^{\prime}$\;
}

\caption{Job Structure Simplifying {\cite{Nagarajan13a}}}\label{Algo-Job-Structure-Simplifying}
\end{algorithm}

\subsection{The Online Learning Algorithm}
\label{sec.learning}

\begin{algorithm*}[t]
	\SetKwInOut{Input}{Input}
	\SetKwInOut{Output}{Output}
	
	\Input{a set $\mathcal{P}$ of $n$ policies, each indexed by $\pi\in\{1, 2,$ $\cdots,$ $n\}$; the set $\mathcal{J}_{t}$ of jobs that arrive at $t$; }

	\BlankLine

    \tcc{\footnotesize{as the time horizon expands, the algorithm operates as follows at every moment $t$}}
	
    \If{$t=0$}{

        $\kappa\leftarrow 1$\tcp*{\footnotesize{$\kappa$ is used to track the number of times updating the weight distribution}}

	    initialize the weight vector of policies: $w_{\kappa}=\{w_{\kappa, 1}, \cdots, w_{\kappa, n}\}=\left\{\frac{1}{n}, \cdots, \frac{1}{n}\right\}$\;
    }

            $\mathcal{J}_{t}^{\prime}\leftarrow \mathcal{J}_{t}$\;

            \While{$\mathcal{J}_{t}^{\prime}\neq \emptyset$}{

			    Get a job $j$ from $\mathcal{J}_{t}^{\prime}$\;

                Call Algorithm~\ref{Algo-Job-Structure-Simplifying} and we get a job $j^{\prime}$ with a chain precedence constraint\;

			    Pick a policy $\pi_{j}=\pi$ with a probability $w_{\kappa,\pi}$\;

                Apply the parametric policy $\pi_{j}$ to $j^{\prime}$ and execute the job $j^{\prime}$ by Algorithm~\ref{Algo-instance-allocation-process}\;

                $\mathcal{J}_{t}^{\prime}\leftarrow \mathcal{J}_{t}^{\prime}-\{j\}$\;

            }

	        	\If{$t > d$}{
		
                   $\mathcal{J}_{t-d}^{\prime}\leftarrow \mathcal{J}_{t-d}$\;

			       \While{$\mathcal{J}_{t-d}^{\prime}\neq \emptyset$}{
				
				        Get a job $j$ from $\mathcal{J}_{t-d}^{\prime}$\;

                        Compute the cost of completing $j$ in the period $[a_{j}, d_{j}]$ under every policy $\pi\in\mathcal{P}$, denoted by $c_{j}(\pi)$\;

				       $\eta_{t}\leftarrow \sqrt{\frac{2\log{n}}{d(t-d)}}$\;
				
			        	\For{$\pi \leftarrow 1$ \KwTo $n$}{

		        			$w_{\kappa+1,\pi}^{\prime}\leftarrow w_{\kappa,\pi}\exp^{-\eta_{t}c_{j}(\pi)}$\;
				        }
				
        				\For{$\pi \leftarrow 1$ \KwTo $n$}{
		        			$w_{\kappa+1,\pi}\leftarrow \frac{w_{\kappa+1,\pi}^{\prime}}{\sum_{i=1}^{n}{w_{\kappa+1, i}^{\prime}}}$\;
				        }

                       $\kappa\leftarrow \kappa+1$\;
				
				       $\mathcal{J}_{t-d}^{\prime}\leftarrow \mathcal{J}_{t-d}^{\prime}-\{j\}$\;

		         	}
	 	}
	\caption{OptiLearning\label{Regret}}
\end{algorithm*}

In this subsection, we describe TOLA formally. There is a set of $n$ parametric policies $\mathcal{P}$. Each policy is represented as $\{\beta_{0},$ $\beta$, $b\}$ and indexed by $\pi=1, 2, \cdots$. Jobs arrive over time and constitute a set of jobs $\mathcal{J}$, indexed by $j=1, 2, \cdots$. Let $d$ $=$ $\max_{j\in\mathcal{J}}\{d_{j}-a_{j}\}$, {\em i.e.}, the maximum relative deadline of all jobs. Let $\mathcal{J}_{t}\subseteq \mathcal{J}$ denote all jobs $j$ arriving at time $t$, {\em i.e.}, $a_{j}$ $=$ $t$. If $\mathcal{J}_{t}\neq \emptyset$, there are jobs arriving at time $t$. There is a weight distribution $w$ over the $n$ policies whose initial value is $\{1/n, \cdots, 1/n\}$. At every moment $t$, TOLA operates as follows to determine the allocation of $\mathcal{J}_{t}$ and update the weight distribution $w$:
\begin{itemize}
\item If $\mathcal{J}_{t}\neq \emptyset$, it randomly picks for each job $j\in\mathcal{J}_{t}$ a policy $\pi_{j}$ from $\mathcal{P}$ according to the current distribution $w$ and bases the allocation of instances to $j$ on the policy $\pi_{j}$; the resulting cost of completing $j$ is denoted by $c_{j}(\pi_{j})$ (lines 4-10).
\item The distribution $w$ is updated when $t\geq d$ and $\mathcal{J}_{t-d}\neq \emptyset$. At such moment $t$, we have the knowledge of spot prices in $[t-d, t]$ and can derive the cost of completing each job $j^{\prime}\in\mathcal{J}_{t-d}$ under every policy $\pi\in\mathcal{P}$, denoted by $c_{j}(\pi)$ (lines 14-15); the distribution is updated such that the lower-cost ({\em resp.} higher-cost) polices of this job are re-assigned the enlarged ({\em resp.} reduced) weights (lines 16-20).
\end{itemize}
Thus, as time goes by and more and more jobs are processed, the most cost-effective policies of $\mathcal{P}$ will be identified gradually, {\em i.e.}, the ones with the highest weights.

As $t$ becomes larger and larger, TOLA will choose for every arriving job $j$ the most cost-effective policy $\pi_{j}$ at a high probability. Finally, the actual total cost of completing all jobs is close to the cost of completing all jobs under a specific policy $\pi^{*}\in\mathcal{P}$ that generates the lowest total cost, {\em i.e.},
\begin{align*}
\pi^{*} = \argmin\limits_{\pi\in \mathcal{P}}\left\{\sum\limits_{t\in [d, T]}\sum\limits_{j\in\mathcal{J}_{t}}{c_{j}(\pi)}\right\}.
\end{align*}
This is formalized as the following proposition. Let $N^{\prime}=|\cup_{t=d}^{T}{\mathcal{J}_{t}}|$, {\em i.e.}, the number of all jobs that arrive in $[d, T]$, and, as proved in \cite{Jain14}, we have that
\begin{proposition}[{{\cite[Theorem 1]{Jain14}}}]
\label{proposi-online-learning}
For all $\delta\in (0, 1)$, it holds with a probability at least $1-\delta$ over the random of online learning that
\begin{center}
$\frac{\sum\limits_{t\in [d, T]}\sum\limits_{j\in\mathcal{J}_{t}}{c_{j}(\pi_{j})}-\sum\limits_{t\in [d, T]}\sum\limits_{j\in\mathcal{J}_{t}}{c_{j}(\pi^{\ast})}}{N^{\prime}} \leq 9\sqrt{\frac{2d\log{(n/\delta)}}{N^{\prime}}}$.
\end{center}
\end{proposition}

\end{document}